\documentclass{article}

\usepackage{amsmath,amssymb,pifont}
\usepackage{multicol}
\usepackage{amstext}
\usepackage{amsthm}
\usepackage{multirow}
\usepackage{booktabs}
\usepackage[skip=0pt]{subcaption}
\usepackage{times}
\usepackage{lipsum}
\usepackage[shortlabels]{enumitem}
\usepackage{cancel}
\usepackage{wrapfig}
\usepackage{array}
\usepackage{siunitx}
\usepackage{csvsimple}
\usepackage[multidot]{grffile}
\usepackage{bbm}
\usepackage{dblfloatfix}
\usepackage{hyperref}
\usepackage{makecell}
\usepackage{bbm, dsfont}
\usepackage{mathtools}
\usepackage[dvipsnames]{xcolor}
\usepackage{comment}
\usepackage[capitalise]{cleveref}
\usepackage{float}
\usepackage{caption}

\usepackage{geometry}
\usepackage{mathpazo}
\usepackage{enumitem}

\usepackage[multiple]{footmisc}
\usepackage{mathrsfs}
\usepackage{todonotes}
\usepackage{tikz}
\usepackage{hyperref}
\usepackage{cleveref}
\usepackage{algpseudocode,algorithm,algorithmicx}
\usepackage{xfrac}

\usepackage{graphicx} 
\usepackage{color}

\usepackage{array}
\usepackage{amssymb}
\usepackage{amsmath}
\usepackage{xspace}
\usepackage{fancyhdr}
\usepackage{comment}

\newcommand{\agnoteinline}[1]{\todo[color=purple!30,inline]{AG: #1}}

\hypersetup{final}

\newcommand{\boldzero}{\ensuremath{\boldsymbol{0}}}

\newcommand{\bfA}{\ensuremath{\mathbf{A}}}

\newcommand{\bfC}{\ensuremath{\mathbf{C}}}

\newcommand{\bfG}{\ensuremath{\mathbf{G}}}

\newcommand{\bfI}{\ensuremath{\mathbf{I}}}

\newcommand{\bfZ}{\ensuremath{\mathbf{Z}}}

\newcommand{\bfa}{\ensuremath{\mathbf{a}}}
\newcommand{\bfb}{\ensuremath{\mathbf{b}}}
\newcommand{\bfc}{\ensuremath{\mathbf{c}}}

\newcommand{\bfg}{\ensuremath{\mathbf{g}}}

\newcommand{\bfw}{\ensuremath{\mathbf{w}}}
\newcommand{\bfx}{\ensuremath{\mathbf{x}}}
\newcommand{\bfy}{\ensuremath{\mathbf{y}}}
\newcommand{\bfz}{\ensuremath{\mathbf{z}}}

\newcommand{\calA}{\ensuremath{\mathcal{A}}}

\newcommand{\calD}{\ensuremath{\mathcal{D}}}

\newcommand{\calM}{\ensuremath{\mathcal{M}}}
\newcommand{\calN}{\ensuremath{\mathcal{N}}}
\newcommand{\calO}{\ensuremath{\mathcal{O}}}

\newcommand{\calV}{\ensuremath{\mathcal{V}}}

\renewcommand{\Pr}{\mathop{\mathbf{Pr}}}

\newcommand{\E}{\mathop{\mathbf{E}}}

\newtheorem{lemma}{Lemma}[section]
\newtheorem{theorem}[lemma]{Theorem}

\newtheorem{corollary}[lemma]{Corollary}

\newtheorem{definition}[lemma]{Definition}
\newtheorem{conjecture}[lemma]{Conjecture}

\makeatletter
\newcommand{\vast}{\bBigg@{4}}
\newcommand{\Vast}{\bBigg@{5}}
\makeatother

\newcommand{\ex}[2]{{\ifx&#1& \mathbb{E} \else
\underset{#1}{\mathbb{E}} \fi \left[#2\right]}}
\newcommand{\pr}[2]{{\ifx&#1& \mathbb{P} \else
\underset{#1}{\mathbb{P}} \fi \left[#2\right]}}

\newcommand{\Var}[1]{\ensuremath{\mathbf{Var}\left(#1\right)}}

\newcommand{\ltwo}[1]{\left\|#1\right\|_2}

\usepackage{ulem}

\DeclarePairedDelimiterX{\infdivx}[2]{(}{)}{%
  #1\;\delimsize\|\;#2%
}

\newcommand{\mypar}[1]{\smallskip
	\noindent{\textbf{{#1}:}}}
	
\renewcommand{\epsilon}{\varepsilon}

\renewcommand{\tilde}{\widetilde}

\newcommand{\node}{\texttt{node}}

\newcommand{\clip}[2]{{\sf clip}\left(#1,#2\right)}

\setlist{nolistsep}
\setlist[itemize]{noitemsep, topsep=0pt}

\setlist{nolistsep}
\setlist[itemize]{noitemsep, topsep=0pt}

\renewcommand{\E}{\mathbb{E}}

\newcommand{\Bern}{\text{Bern}}
\newcommand{\Unif}{\text{Uniform}}
\newcommand{\blue}[1]{\color{blue}#1}
\usepackage{fullpage}
\usepackage[numbers, sort, comma, square]{natbib}

\begin{document}
\title{Privacy Amplification for BandMF via $b$-Min-Sep Subsampling}
\author{
Andy Dong\thanks{\texttt{dxa@stanford.edu}, Stanford University}
\and
Arun Ganesh\thanks{\texttt{arunganesh@google.com}, Google Research}
}
\maketitle


\begin{abstract}
We study privacy amplification for BandMF, i.e., DP-SGD with correlated noise across iterations via a banded correlation matrix. We propose $b$-min-sep subsampling, a new subsampling scheme that generalizes Poisson and balls-in-bins subsampling, extends prior practical batching strategies for BandMF, and enables stronger privacy amplification than cyclic Poisson while preserving the structural properties needed for analysis. We give a near-exact privacy analysis using Monte Carlo accounting, based on a dynamic program that leverages the Markovian structure in the subsampling procedure. We show that $b$-min-sep matches cyclic Poisson subsampling in the high noise regime and achieves strictly better guarantees in the mid-to-low noise regime, with experimental results that bolster our claims. We further show that unlike previous BandMF subsampling schemes, our $b$-min-sep subsampling naturally extends to the multi-attribution user-level privacy setting.
\end{abstract}

\section{Introduction}

Differentially private stochastic gradient descent (DP-SGD) is the de facto algorithm for training models with provable privacy guarantees. Compared to standard SGD, DP-SGD clips per-example gradients and adds noise to the aggregated batch gradient. There are several complementary techniques for improving DP-SGD. Representative examples include (1) privacy amplification through subsampling \cite{DP-DL}, which exploits randomness in batching to achieve similar privacy guarantees as deterministic batching with less noise, (2) leveraging algorithmic randomness in the design of the training process \cite{chen2023privacy, dong2025leveraging}, and (3) DP-MF \cite{denisov2022improved}, which correlates the noise added to batch gradients across iterations so that noise injected in one iteration is partially canceled in subsequent iterations.

While DP-MF was initially proposed to achieve good utility in settings where privacy amplification is infeasible, DP-SGD can simultaneously benefit from privacy amplification and DP-MF. However, the standard analysis of privacy amplification uses composition theorems which rely on independence across iterations — an assumption violated by correlated noise in DP-MF. In turn, new analyses and new sampling schemes which are more amenable to analysis are necessary to effectively combine privacy amplification and DP-MF.

In this work, we introduce \textit{$b$-min-sep subsampling}, a simple yet powerful subsampling scheme that reconciles privacy amplification with DP-MF. It generalizes Poisson and balls-in-bins subsampling, two standard sampling schemes for DP-SGD. We focus on BandMF, a variant of DP-MF of interest. For BandMF, 
$b$-min-sep subsampling preserves the local participation constraint required for tractable privacy analysis, while avoiding the additional restrictions imposed by cyclic Poisson subsampling, the scheme originally proposed for BandMF. As a result, it enables stronger privacy amplification.

Crucially, we give a \textit{near-exact} privacy analysis for BandMF with $b$-min-sep subsampling (up to slack due to Monte Carlo accounting). We prove that $b$-min-sep matches cyclic Poisson subsampling in the high-noise regime, strictly improves upon it in the low-noise regime, and therefore Pareto dominates prior approaches in ground-truth privacy. Empirically, we observe consistent gains in utility at moderate-to-large privacy budgets.

Finally, unlike prior amplification schemes for BandMF, $b$-min-sep naturally extends to the multi-attribution user-level privacy setting, where examples may be shared by multiple users. Previous methods either break down or require prohibitive restrictions that degrade utility too much. Taken together, our results position $b$-min-sep subsampling as a unified, principled, and practical foundation for privacy amplification in correlated-noise mechanisms.

\section{Background}

Before discussing our results in further detail, we review the necessary background on DP-SGD, DP-MF, and privacy amplification.

\subsection{Differential Privacy}

The canonical $(\varepsilon, \delta)$-DP definition \cite{DMNS} has many different forms. We use the following based on the zero-out adjacency (as is now common in the DP training literature \cite{ponomareva23dpfy, pillutla2025correlatednoisemechanismsdifferentially}) throughout the paper:

\begin{definition}
    Given a data universe $\calD$, we say that two datasets $D, D' \in \calD^n$ are adjacent under the zero-out adjacency, denoted $D \simeq D'$ if $D'$ is the same as $D$ except with exactly one element replaced with $\bot$, a special element whose gradients are zero everywhere, i.e. $\nabla \ell(\theta; \bot) = \boldzero$ for all $\theta$.
\end{definition}

\begin{definition}
    Given two distributions $P, Q$ over $\calO$, the $\alpha$-hockey stick divergence between them is
    \[H_\alpha(P, Q) := \int_\calO \max\{P(y) - \alpha Q(y), 0\}\ dy = \mathbb{E}_{y \sim P} \left[\max\left\{1 - \alpha \frac{Q(y)}{P(y)}, 0\right\}\right].\]

    Given a mechanism $\calM: \calD^* \rightarrow \Delta(\calO)$ where $\Delta(\calO)$ denotes the set of distributions over the set $\calO$, we say $\calM$ satisfies $(\varepsilon, \delta)$-DP (with respect to the zero-out adjacency) if for all $D \simeq D'$, $H_{e^\varepsilon}(\calM(D), \calM(D')) \leq \delta.$
\end{definition}
We note that the above is completely equivalent to the standard $(\epsilon, \delta)$-DP definition (and just presented in a different way) but is more convenient to work with for Monte Carlo accounting, which is the technique we will use.

\subsection{DP-SGD}

DP-SGD \cite{song2013stochastic, BST14, DP-DL} is a now-canonical method for private machine learning with $(\varepsilon, \delta)$-DP. We assume a sequence of batches $B_1, B_2, \ldots, B_n$ to be used in DP-SGD is given to us. In iteration $i$, letting $\clip{\theta}{C} := \theta \cdot \frac{C}{\max\{\ltwo{\theta}, C\}}$, given clip norm $C$ and noise multiplier $\sigma$ we compute the following noisy gradient sum on the batch $B_i$ and the current model $\theta_{i-1}$:

\[\tilde{g}_i = \sum_{d \in B_i} \clip{\nabla \ell(\theta_{i-1}; d)}{C} + \bfz_i, \quad\quad \bfz_i \sim \calN(\boldzero, \sigma^2 C^2 \bfI).\]

$\tilde{g}_i$ can then be passed as a regular gradient would to any underlying first-order optimizer (e.g. SGD or Adam). Without loss of generality, we will implicitly assume $C = 1$ for the rest of the paper.

\subsection{DP-MF}

DP matrix factorization (DP-MF) \cite{denisov2022improved, choquette2022multi, choquette2024amplified} is a variant of DP-SGD that adds correlated noise in each iteration rather than independent noise. Namely, we have an additional hyperparameter $\bfC \in \mathbb{R}^{n \times n}$ referred to as the strategy matrix; as is standard, we will assume throughout the paper $\bfC$ is lower-triangular and non-negative which enables adaptivity of DP-MF and privacy amplification as discussed in \cite{denisov2022improved, choquette2023privacy}. Given noise samples $\bfz_1, \bfz_2, \ldots, \bfz_t \sim \calN(\boldzero, \sigma^2 \bfI)$, let $\bfZ$ denote the matrix formed by stacking these samples, i.e. the $i$th row of $\bfZ$ is $\bfz_i$. Rather than add noise $\bfz_i$ in iteration $i$, we now add noise $(\bfC^{-1} \bfZ)[i, :]$, i.e. the $i$th row of $\bfC^{-1} \bfZ$; the choice to multiply by $\bfC^{-1}$ rather than $\bfC$ will be made clear when we later discuss the privacy analysis of DP-SGD and DP-MF. A typical choice of $\bfC^{-1}$ will have a positive main diagonal and negative below-diagonal entries, which corresponds to adding some noise in one iteration and partially cancelling it out in subsequent iterations. The choice $\bfC = \bfI$ retrieves DP-SGD.

The key advantage of DP-MF is that it trades off per-iteration noise magnitude against structured noise correlations across iterations. As $\bfC$ departs from the identity, correlations reduce the variance of cumulative noise over multiple iterations, which directly controls the noise in model parameters, but require a larger per-iteration noise multiplier $\sigma(\bfC)$ to maintain a fixed privacy guarantee. DP-SGD corresponds to the extreme choice $\bfC=\bfI$, which minimizes per-iteration noise $\sigma(\bfC)$ but accumulates noise rapidly over multiple iterations. DP-MF improves utility by selecting $\bfC$ in an intermediate regime that optimally balances these effects. Equivalently, DP-MF can be viewed as optimizing $\bfC$ to minimize the total effect of noise over all iterations under $(\varepsilon,\delta)$-DP, whereas DP-SGD corresponds to a fixed, non-optimized choice. 

A common measure of the utility of a particular choice of $\bfC$ is the mean squared error (MSE) of its prefix sums, as proposed in \cite{denisov2022improved}. That is, let $\bfA$ be the all-ones lower-triangular matrix. Let $\sigma(\bfC)$ be the noise multiplier necessary for achieving a target privacy guarantee using DP-MF with a given choice of $\bfC$ (which we will discuss how to compute in the next section). $\bfA \bfC^{-1} \bfZ$ is a matrix whose $i$th row is the sum of the first $i$ noises added in DP-MF, so the MSE of the prefix sums is given by $\|\bfA \bfC^{-1}\|_F^2 \sigma(\bfC)^2$.

The choice of $\bfC$ also impacts the efficiency of DP-MF (depending on the time and memory efficiency of computing rows of $\bfC^{-1} \bfZ$). A large number of works \cite{denisov2022improved, choquette2022multi, choquette2024amplified, choquette2023correlated, dvijotham2024efficient, henzinger2023almost, henzinger2024unifying, henzinger25improved, kalinin2024banded, kalinin2025squarerootsoptimalbound, gu2025correlatingcrossiterationnoisedpsgd} propose different methods for choosing $\bfC$ offering different tradeoffs between efficiency, utility, and other factors. For this paper we will focus on banded Toeplitz $\bfC$ \cite{choquette2024amplified}. A (lower triangular) matrix $\bfC$ is $b$-banded if only its first $b$ diagonals are non-zero, i.e. $\bfC_{i,j} \neq 0$ only if $i - b + 1 \leq j \leq i$, and Toeplitz if all its diagonals are constant-valued, i.e. $\bfC_{i,i+k} = \bfC_{i', i'+k}$ for all $i, i', k$. DP-MF using banded Toeplitz matrices is of particular interest because it can be implemented efficiently and with only a small loss in utility compared to more general families of matrices \cite{mckenna2025scaling}, and banded matrices are particularly amenable to privacy amplification \cite{choquette2024amplified, choquette2024near}. We will sometimes use BandMF to refer to DP-MF when using a banded $\bfC$, and $b$-BandMF when the number of bands $b$ is specified.

\subsection{Privacy Accounting and Amplification}

To perform privacy accounting, i.e. determine what $(\varepsilon, \delta)$-DP guarantees are satisfied by a given instantiation of DP-SGD/DP-MF, since the final model is a post-processing of the privatized gradients $\tilde{\bfg}_i$ it suffices to analyze the privacy of releasing the sequence of privatized gradients $\tilde{\bfg}_1, \tilde{\bfg}_2, \ldots, \tilde{\bfg}_n$. Without loss of generality we will focus on privacy accounting for DP-MF \cite{choquette2023privacy}. Let $\bfG$ be the matrix whose $i$th row is the sum of clipped gradients in iteration $i$. The sequence of privatized gradients can equivalently be expressed as $\bfG + \bfC^{-1} \bfZ$. By post-processing, it is equivalent to consider the privacy of the matrix mechanism $\bfC \bfG + \bfZ$. By e.g. Lemma 4.5 of \cite{choquette2023privacy}, it suffices to analyze the (worst-)case where all examples in $D \cap D'$ have zero gradients, and $\bfG$ is one-dimensional and non-adaptive. Equivalently, the worst-case privacy analysis across all pairs of datasets is achieved by the matrix mechanism $\bfC \bfx + \bfz$ where now the noise $\bfz$ is a vector, and $\bfx$ is all-zeros for $D'$ and for $D$ it is $1$ in coordinates corresponding to iterations where the example in $D \setminus D'$ participates, and $0$ elsewhere. When $\bfx$ is deterministic for $D$, the matrix mechanism $\bfC \bfx + \bfz$ is just a Gaussian mechanism and can be easily analyzed using standard tools \cite{balle2018improving}. Generally, as $\bfC$ becomes less diagonal, this Gaussian mechanism's sensitivity increases, requiring larger $\sigma(\bfC)$. 

Privacy amplification (e.g. by sampling) is the observation that forming batches randomly, i.e. $\bfx$ being a random vector instead of a deterministic vector, generally improves the privacy guarantee (or equivalently, reduces the $\sigma$ needed to achieve a target privacy guarantee). A standard sampling scheme enabling privacy amplification of DP-SGD is Poisson subsampling \cite{kasiviswanathan2008what}, where the $i$th batch $B_i$ is formed by independently (across examples and across iterations) including each example in $D$ with probability $p$. For DP-SGD with Poisson subsampling (i.e. $\bfC = \bfI$) we can analyze the privacy of $\bfx + \bfz$ by computing a privacy guarantee for releasing a single coordinate of $\bfx + \bfz$ (which, since this is a single-dimensional output, can be done near-exactly using e.g. numerical integration \cite{doroshenko2022connect}), and then applying a composition theorem which bootstraps the privacy guarantee for a single coordinate into a privacy guarantee for the whole mechanism. However, applying the composition theorem assumes each coordinate is independent of each other, which no longer holds for $\bfC \bfx + \bfz$ if $\bfC$ is not diagonal. 

\subsection{Monte Carlo Accounting}

Monte Carlo accounting \cite{wang2023randomized} is a privacy analysis technique which \cite{choquette2024near} observed is useful for privacy analysis of matrix mechanisms because it bypasses the need for composition theorems. The key idea behind Monte Carlo accounting is that the hockey-stick divergence\footnote{We really want to compute the maximum of $H_{e^\varepsilon}(P, Q)$ and $H_{e^\varepsilon}(Q, P)$; for simplicity we focus on one direction in the paper and it is straightforward to extend our techniques to the other direction.} $\mathbb{E}_{y \sim P}\left[\max\left\{1-e^\varepsilon \frac{Q(y)}{P(y)}, 0\right\}\right]$ may be hard to compute, but if we can efficiently sample $y \sim P$ and compute $Q(y)/P(y)$ we can estimate the hockey stick divergence using Monte Carlo estimation, i.e. taking many samples of $\max\left\{1-e^\varepsilon \frac{Q(y)}{P(y)}, 0\right\}$ and using their average as an estimate of $\delta$.

While using this estimate alone does not satisfy proper privacy guarantees, \cite{wang2023randomized} give an ``Estimate-Verify-Release'' (EVR) framework to turn a Monte Carlo estimate of $\delta$ into a formal privacy guarantee. EVR takes a given mechanism, runs Monte Carlo estimation as a verifier to verify the mechanism satisfies $(\varepsilon, \tau \delta)$-DP ($\tau < 1$), then runs the mechanism only if the verifier succeeds, otherwise returning an empty output. If for any mechanism violating $(\varepsilon, \delta)$-DP, the verifier succeeds with probability at most $\delta$, then the overall EVR framework is $(\varepsilon, \delta)$-DP. Fact 4.2 in \cite{chua25ballsandbins} gives a bound on the failure probability of Monte Carlo estimation as a verifier in terms of the number of Monte Carlo samples and $\delta$. We give a slight variant of EVR in \cref{section:evr} which removes the possibility of an empty output and allows us to combine tuning the noise multiplier with the verification step. Our formal discussion of EVR can be skipped by most readers; to understand the main results in the paper, it suffices to understand that EVR implies that (i) taking samples $y \sim P$ and computing $Q(y)/P(y)$ are sufficient primitives for choosing a noise multiplier that satisfies a formal end-to-end DP guarantee (ii) this procedure results in a slight overestimate of the necessary noise multiplier, which becomes exact in the limit as the number of Monte Carlo samples goes to infinity.

The need to efficiently evaluate the likelihood ratio $Q/P$ is a non-trivial bottleneck to combining Monte Carlo accounting with DP-MF. In Monte Carlo accounting of DP-MF, if $p_\bfx$ gives the probability of sampling a given $\bfx$ value, then we have e.g. $P = \sum_{\bfx} p_\bfx \cdot \calN(\bfC \bfx, \sigma^2 \bfI)$, $Q = \calN(\boldzero, \sigma^2 \bfI)$. For e.g. Poisson subsampling, $\bfx$ can take on $2^n$ values, hence evaluating this sum seemingly takes exponential time. Hence, even with Monte Carlo accounting one needs to be careful in designing the sampling scheme to be amenable to privacy analysis.

\section{Past Work and Our Contributions}

We first reflect on past approaches for combining privacy amplification with DP-MF and where there is room for improvement.

\mypar{Cyclic Poisson subsampling of \cite{choquette2024amplified}} In cyclic Poisson subsampling, we partition the dataset into $b$ subsets $D_1, \ldots D_b$, and in iteration $i$ Poisson sample from $D_{i \pmod b}$\footnote{Abusing notation and letting $i \pmod b = b$ if $i$ is a multiple of $b$.} to form the $i$-th batch. For banded Toeplitz $\bfC$, without loss of generality\footnote{As the worst case for privacy is that the example in $D \setminus D'$ is $D_1$.}, the privacy analysis reduces to choosing $\bfx$ such that $\bfx_1, \bfx_{b+1}, \bfx_{2b+1} \ldots$ are each sampled from $\Bern(p)$ and all other entries of $\bfx$ are 0. When combined with $b$-BandMF, the coordinates of the output $\bfC \bfx + \bfz$ that are affected by each of $\bfx_1, \bfx_{b+1}, \bfx_{2b+1} \ldots$ are disjoint sets. In turn, applying cyclic Poisson subsampling to $b$-BandMF can be analyzed similarly to DP-SGD with Poisson subsampling, as it is the composition of $n/b$ independent events. However, since we are sampling from datasets that are $b$ times smaller, we need to use a sampling probability $b$ times larger than DP-SGD's to achieve the same expected batch size, which leads to worse amplification than DP-SGD. An equivalent viewpoint is that because we are sampling a same-size batch from a smaller dataset in each iteration, there is less randomness in the sampling process and thus less potential for amplification.

\mypar{Conditional composition of \cite{choquette2023privacy}} Conditional composition attempts to use composition theorems to analyze amplified DP-MF by (1) finding a mechanism for each iteration such that with high probability, the mechanism outputting the $i$th coordinate of of $\bfC \bfx + \bfz$ conditioned on the first $i-1$ coordinates is dominated by this mechanism, and then (2) composing the privacy guarantees of the dominating mechanisms to arrive at privacy guarantee for the overall mechanism $\bfC$, absorbing the high probability guarantee into $\delta$. While conditional composition allows for arbitrary $\bfC$ and flexible sampling schemes, the fact that it effectively conditions on near-worst case events to build its privacy guarantee means there can be a lot of slack in the analysis. Since the amount of slack introduced is too large in our experimental settings, we do not consider conditional composition as a competitive baseline for the remainder of the paper.

\mypar{Balls-in-bins subsampling of \cite{choquette2024near}} Balls-in-bins subsampling mimics shuffling but keeps the participation of each example independent from each other. Given an epoch length $T$, each example is independently assigned a uniformly random index $i$ in $[T]$, and then participates in iterations $\{i, T+i, 2T+i \ldots\}$. \cite{choquette2024near} observed that balls-in-bins is amenable to Monte Carlo accounting because there are only $T$ possible values of $\bfx$, hence the likelihood ratio $Q/P$ can be evaluated efficiently. While in their experiments balls-in-bins seems to empirically outperform cyclic Poisson subsampling, it uses less randomness than Poisson subsampling\footnote{Balls-in-bins is equivalent to Poisson sampling conditioned on examples participating once per epoch; this conditioning reduces the randomness.}, which suggests that Poisson subsampling or similar sampling schemes could still achieve even better privacy amplification.

\mypar{Our contributions} The goal of our work is to improve the combination of privacy amplification and DP-MF. Our contributions are as follows:

\begin{itemize}
    \item In \cref{section:bms}, we propose $b$-min-sep subsampling, and demonstrate that it is a generalization and unification of Poisson subsampling and balls-in-bins subsampling. We also show that it achieves strictly stronger privacy guarantees than cyclic Poisson in the limit as $\sigma \rightarrow 0$, and approaches the privacy analysis of cyclic Poisson in the limit as $\sigma \rightarrow \infty$. Interpolating between these extremes suggests that $b$-min-sep subsampling Pareto dominates all past practical approaches for privacy amplification of BandMF.
    \item In \cref{section:montecarlo}, we give an efficient implementation of Monte Carlo accounting for BandMF with $b$-min-sep subsampling. The key technical challenge is that $b$-min-sep subsampling can generate exponentially many different values of $\bfx$, which means directly computing $Q/P$ via summation over $\bfx$ is infeasible. Our main insight is that the summation can be subverted by instead solving a dynamic program.
    \item In \cref{section:multi-attribution} we consider the multi-attribution setting, a generalization of user-level DP introduced in \cite{ganesh2025s} where each example can be attributed to multiple users. While cyclic Poisson and balls-in-bins are inherently challenging to apply to the multi-attribution setting, we demonstrate that $b$-min-sep subsampling and its analysis easily extend to the multi-attribution setting. 
    \item In \cref{section:experiments} we conduct experiments showing that both in terms of MSE and the test accuracy of a standard CIFAR benchmark, $b$-min-sep subsampling offers improvements over balls-in-bins and cyclic Poisson in the high-epsilon regime, however cyclic Poisson tends to perform better in the low $\varepsilon$ regime. However, we also demonstrate that the improvement of cyclic Poisson is only due to the slack in Monte Carlo accounting, and $b$-min-sep subsampling outperforms cyclic Poisson at all $\varepsilon$ if we optimistically assume no error in Monte Carlo estimation. We also conduct experiments showing that $b$-min-sep subsampling's unique ability to utilize both privacy amplification and correlated noise in the multi-attribution setting realizes in significant improvements for fine-tuning on the arXiv dataset.
\end{itemize}

\section{$b$-min-sep subsampling}\label{section:bms}

We now introduce $b$-min-sep subsampling, a new subsampling scheme designed to enforce the
$b$-min-sep participation constraint. The $b$-min-sep property itself, that no example participates in iterations $i \neq j$ with $|i-j|<b$, was previously introduced as a structural condition for analyzing
$b$-BandMF \cite{choquette2024amplified}. When this property holds, repeated participations of the same example do not affect overlapping coordinates of $\bfC \bfx + \bfz$, which significantly simplifies privacy analysis.

Our proposed $b$-min-sep subsampling\footnote{The $b$-min-sep subsampling scheme should not be confused with the $b$-min-sep property introduced in prior work; we adopt the same name because the scheme enforces this property by construction.} (\cref{fig:bms_basic}) is a surprisingly simple and natural procedure that enforces this property: perform Poisson subsampling in each iteration, but exclude any example that participated in the previous $b-1$ iterations.

\begin{figure}[H]
\begin{algorithm}[H]
\caption{$b$-min-sep subsampling}
\label{fig:bms_basic}
\textbf{Parameters:} Sampling probability $p$, minimum separation parameter $b$, dataset $D = \{e_1, e_2, \ldots e_m\}$, total iterations $n$.
\begin{algorithmic}[1]
\For{$i \in [n]$}
\State $D^{exclude} = \bigcup_{j = \max\{1, \: i-b+1\}}^{i-1} B_j$
\State $B_i \leftarrow$ include each element of $D \setminus D^{exclude}$ independently with probability $p$.
\EndFor
\State \Return batches $B_1, B_2, \ldots B_n$.
\end{algorithmic}
\end{algorithm}
\end{figure}

An alternative way to view $b$-min-sep subsampling is through the lens of a Markov chain that tracks the availability state of a single example across iterations. Consider states $\{0, 1, \ldots, b-1\}$, where state $0$ means the example is available for Poisson subsampling, and state $i > 0$ means it is barred from participation for the next $i$ iterations, with transition probabilities given in Figure~\ref{fig:bms_markov}. The example participates when it transitions to state $b-1$. The expected return time to state $b-1$ is $(b-1) + 1/p$, so given a dataset $D$ and a target expected batch size $\bar{B} = |D|p_0$, we should set $1/p_0 = 1/p + (b-1)$, or equivalently $p = \frac{p_0}{1-p_0(b-1)}$.

\begin{figure}[h]
\centering
\begin{tikzpicture}[>=stealth, node distance=2.1cm]

\node[circle,draw] (s0) {$0$};
\node[circle,draw,right of=s0] (s1) {$1$};
\node[circle,draw,right of=s1] (s2) {$2$};
\node[right of=s2,node distance=1.2cm] (dots) {$\cdots$};
\node[circle,draw,right of=dots,node distance=1.6cm] (sb) {$b\!-\!1$};

\draw[->] (s0) edge[loop above] node {$1-p$} (s0);
\draw[->] (s0) edge[bend right=20] node[above] {$p$} (sb);

\draw[->] (s1) -- node[above] {$1$} (s0);
\draw[->] (s2) -- node[above] {$1$} (s1);
\draw[->] (sb) -- node[above] {$1$} (dots);
\draw[->] (dots) -- node[above] {$1$} (s2);

\end{tikzpicture}
\caption{Markov chain governing the availability state of a single example in $b$-min-sep subsampling.}
\label{fig:bms_markov}
\end{figure}

A minor inconvenience of Algorithm~\ref{fig:bms_basic} is that all examples start in state $0$, so the expected batch size varies greatly in the starting iterations, until the Markov chain for each example approaches stationary distribution and then the batch size stays roughly constant. The varying batch size may cause inconvenience in implementation or hinder model convergence. To avoid it, we propose Algorithm~\ref{fig:bms_ws}, a ``warm-start'' variant of $b$-min-sep subsampling that starts each example independently from the stationary distribution so the expected batch size is always constant. Equivalently, we assign a ``virtual history'' to each example according to the stationary distribution.

\begin{figure}[H]
\begin{algorithm}[H]
\caption{``Warm-start'' $b$-min-sep subsampling}
\label{fig:bms_ws}
\textbf{Parameters:} Sampling probability $p$, minimum separation parameter $b$, dataset $D = \{e_1, e_2, \ldots e_m\}$, total iterations $n$.
\begin{algorithmic}[1]
\For{$j \in [m]$}
\State $s_j = 0 \text{ w.p. }\frac{1}{1 + (b-1)p}, \text{otherwise } s_j \sim \Unif(\{1, 2, \ldots, b-1\})$
\EndFor
\For{$i \in [b-1]$}
\State $B_{1-i} \leftarrow \{e_j: s_j = i\}$ (virtual history)
\EndFor
\For{$i \in [n]$}
\State $D^{exclude} = \bigcup_{j = i-b+1}^{i-1} B_j$
\State $B_i \leftarrow$ include each element of $D \setminus D^{exclude}$ independently with probability $p$.
\EndFor
\State \Return batches $B_1, B_2, \ldots B_n$.
\end{algorithmic}
\end{algorithm}
\end{figure}

\cref{fig:bms_ws} can be viewed as an interpolation between Poisson subsampling and balls-in-bins subsampling as defined in \cite{chua25ballsandbins, choquette2024near}. When $b = 1$, we retrieve Poisson subsampling. When $p = 1$ and $p_0 = 1 / b$, the warm-start in \cref{fig:bms_ws} reduces to forming $b$ batches up front by assigning each example to one of the $b$ batches uniformly at random, and then \cref{fig:bms_ws} will cycle through these same batches, exactly as balls-in-bins subsampling does. So, since we will give an exact privacy analysis for $b$-min-sep subsampling, our overall framework can be viewed as a generalization and unification of both privacy amplification via Poisson subsampling and privacy amplification via balls-in-bins subsampling. 

\subsection{Comparison to Cyclic Poisson subsampling}

We give a thorough comparison between $b$-min-sep and cyclic Poisson subsampling. Intuitively, the advantage of $b$-min-sep subsampling is that it still satisfies the $b$-min-sep property (which as we will show in \cref{section:montecarlo} allows an exact privacy analysis when combined with $b$-BandMF), but it is able to sample from a larger pool of examples in each iteration and hence use a smaller sampling probability than cyclic Poisson subsampling, and smaller sampling probabilities lead to better privacy amplification. One way to see this is that cyclic Poisson subsampling excludes all examples which \textit{could} have participated in the past $b-1$ iterations, even if they didn't, whereas $b$-min-sep subsampling only excludes the examples which \textit{did} participate in the past $b-1$ iterations.

To further quantify the gain: if we use an expected batch size of $B$, cyclic Poisson subsampling will only be able to sample from $|D|/b$ examples each iteration whereas $b$-min-sep subsampling has at least $|D| - (b-1)B$ examples available for sampling (in expectation). Since $B \leq |D| / b$ is needed for both $b$-min-sep and an expected batch size of $B$ to be feasible and this inequality implies $|D|/b \leq |D| - (b-1)B$, this means $b$-min-sep subsampling always has at least as many examples available for sampling as cyclic Poisson subsampling. However, if $bB$ is much smaller than $D$, $b$-min-sep subsampling will have $\approx b$ times as many examples available for sampling in each iteration. This enables us to use smaller per-example sampling probabilities and hence should lend itself to better privacy amplification.

While it is challenging to formally show that $b$-min-sep subsampling leads to better privacy guarantees than cyclic Poisson subsampling in all instantiations of $b$-BandMF, we can consider two extreme cases that are easier to asymptotically compare the two methods in: the limit as $\sigma \to 0$ and the limit as $\sigma \to \infty$. For a meaningful comparison, we use different subsampling probabilities so that the expected batch size per iteration is the same for both methods. We fix an average participation rate of $p_0$. Cyclic Poisson subsampling uses a subsampling probability of $p_{\mathrm{cyclic}} = b p_0$ (whenever we are in an iteration where the example can participate). For $b$-min-sep subsampling, as previously derived we have $p_{\mathrm{bMS}} = \frac{1}{1/p_0 - (b-1)} = \frac{p_0}{1 - p_0 (b-1)}$.

\mypar{Limit as $\sigma \to 0$} We first consider the small-noise regime $\sigma \to 0$, which corresponds to high utility and large privacy budget. In this limit, the random variable $\bfy = \bfC \bfx + \bfz$ is dominated by $\bfC \bfx$. The privacy loss random variable can be written as
\[
L(\bfy)
= \ln \left( \sum_{\bfx} p_{\bfx} \exp\!\left( \frac{2 \langle \bfC \bfx, \bfy \rangle - \|\bfC \bfx\|_2^2}{2\sigma^2} \right) \right),
\]
where $\bfx$ ranges over all valid binary participation vectors. As $\sigma \to 0$, standard Laplace-type arguments imply that the log-sum-exp expression is dominated by the single term with the largest exponent, since all other terms are exponentially suppressed by $\exp \left(-1/\sigma^2 \right)$ if we factor out the largest term. Equivalently, with overwhelming probability, $\bfy$ lies within $O(\sigma)$ of exactly one mode $\bfC \bfx$, and contributions from all other modes to $L(\bfy)$ are negligible.

Consequently, except on an exponentially small event, the privacy loss is well approximated by
\[
L(\bfy) \approx \frac{\|\bfC \bfx\|_2^2}{2\sigma^2} + \ln p_{\bfx}
\quad \text{for the $\bfx$ such that } \bfy \approx \bfC \bfx.
\]
Thus, up to lower-order terms in $\sigma$, in the limit as $\sigma \to 0$ an $(\varepsilon, \delta)$-DP guarantee is equivalent to a $1 - \delta$ probability tail bound on $\|\bfC \bfx\|_2^2$. Under both subsampling schemes, any valid participation vector $\bfx$ has some $1$ entries separated by at least $b$ $0$s. Since $\bfC$ is a $b$-banded Toeplitz matrix, this implies
\[
\|\bfC \bfx\|_2^2 \approx \|\bfc_1\|_2^2 \cdot \|\bfx\|_2^2,
\]
where $\bfc_1$ is the first column of $\bfC$ and $\|\bfx\|_2^2$ is exactly the total number of participations of a given example. The approximation here is only due to the fact that the last $b-1$ columns of $\bfC$ have a different norm, and when $n \gg b$ this approximation is minor. Therefore, in the $\sigma \to 0$ regime, comparing the DP guarantees for the sampling schemes reduces to comparing tail bounds on their number of participations. For a fixed expected number of participations, a tail bound is controlled to first order by the variance of the number of participations.

Under cyclic Poisson subsampling, an example has $t/b$ independent opportunities to participate over $t$ iterations, each with probability $p_{\mathrm{cyclic}} = b p_0$. The participation count is therefore binomial with variance
\[
\mathrm{Var}_{\mathrm{cyclic}}
= \frac{t}{b} \, p_{\mathrm{cyclic}} (1 - p_{\mathrm{cyclic}})
= n \, p_0 (1 - b p_0).
\]

On the other hand, under $b$-min-sep subsampling, participations form a renewal process with i.i.d.\ inter-arrival times given by a geometric random variable with success probability $p_{\mathrm{bMS}}$, shifted by $(b-1)$. The inter-arrival time has mean
\[
\mu_{\mathrm{interarrival}} = \frac{1}{p_{\mathrm{bMS}}} + (b-1),
\]
and variance
\[
\sigma_{\mathrm{interarrival}}^2 = \frac{1 - p_{\mathrm{bMS}}}{p_{\mathrm{bMS}}^2}.
\]
Standard renewal theory implies that the asymptotic variance\footnote{We justify the use of the asymptotic variance expression above by noting that modern deep neural network training typically involves a very large number of iterations $n$. In this regime, the number of participations accumulated over training is well described by its long-run behavior. The asymptotic variance arises from classical renewal central limit theorems, which state that, for a renewal process with finite inter-arrival mean and variance, the centered and normalized count of arrivals converges in distribution to a Gaussian as $t \to \infty$, with variance growing linearly in $t$ at rate $\sigma_{\mathrm{interarrival}}^2 / \mu_{\mathrm{interarrival}}^3$. Since the inter-arrival distribution induced by $b$-min-sep has finite moments and $n$ is large in typical training pipelines, this asymptotic approximation provides an accurate characterization of the distribution of the number of participations for the purposes of our comparison.} of the number of participations over $n$ iterations is
\[
\mathrm{Var}_{\mathrm{bMS}}
= \frac{\sigma_{\mathrm{interarrival}}^2}{\mu_{\mathrm{interarrival}}^3} \, n
= \frac{\frac{1 - p_{\mathrm{bMS}}}{p_{\mathrm{bMS}}^2}}{\left( \frac{1}{p_{\mathrm{bMS}}} + (b-1) \right)^3} \, n
= \frac{p_{\mathrm{bMS}}(1 - p_{\mathrm{bMS}})}{\bigl(1 + p_{\mathrm{bMS}} (b-1)\bigr)^3} \, n.
\]

Recalling that $p_{\mathrm{bMS}} = \frac{p_0}{1 - p_0 (b-1)}$, we substitute to obtain
\[
\mathrm{Var}_{\mathrm{bMS}}
= \frac{\frac{p_0}{1 - p_0 (b-1)} \cdot \frac{1 - b p_0}{1 - p_0 (b-1)}}{\left( \frac{1}{1 - p_0 (b-1)} \right)^3} \, n
= n \, p_0 (1 - b p_0)\bigl(1 - p_0 (b-1)\bigr) = \mathrm{Var}_{\mathrm{cyclic}}\bigl(1 - p_0 (b-1)\bigr).
\]
Since $p_0 \in (0, 1/b)$ and $b \geq 1$, we have $1 - p_0 (b-1) < 1$, and therefore $\mathrm{Var}_{\mathrm{bMS}} < \mathrm{Var}_{\mathrm{cyclic}}$, with equality only when $b = 1$, which is what we expect since the two subsampling schemes coincide. This implies that $b$-min-sep subsampling achieves a better tail bound for the number of participations $\|\bfx\|_2^2$, and hence better $(\varepsilon, \delta)$-DP guarantees in the limit as $\sigma \to 0$.

\vspace{1em}
\mypar{Limit as $\sigma \to \infty$} We analyze the high-noise regime $\sigma \to \infty$ by examining the asymptotic behavior of the privacy loss random variable. Under the adjacent dataset lacking the differing example, the output is pure Gaussian noise $\bfy = \bfz \sim \mathcal{N}(\boldzero, \sigma^2 \bfI)$. We reparameterize the noise as $\bfz = \sigma \bfw$, where $\bfw \sim \mathcal{N}(\boldzero, \bfI)$. Under the dataset containing the differing example, the output is $\bfy = \bfC \bfx + \bfz$, where $\bfx$ is the participation vector whose distribution is determined by the subsampling scheme. The likelihood ratio is given by 
\[\frac{P(\bfy)}{Q(\bfy)} = \mathbb{E}_{\bfx \sim P}\left[ \exp\left( \frac{\langle \bfC\bfx, \bfw \rangle}{\sigma} - \frac{\|\bfC\bfx\|_2^2}{2\sigma^2} \right) \right].\]
Taking the Taylor expansion of the exponential function, $\exp(u) = 1 + u + u^2/2 + \mathcal{O}(u^3)$, yields
\[ \frac{P(\bfy)}{Q(\bfy)} = \mathbb{E}_{\bfx \sim P}\left[ 1 + \left( \frac{\langle \bfC\bfx, \bfw \rangle}{\sigma} - \frac{\|\bfC\bfx\|_2^2}{2\sigma^2} \right) + \frac{1}{2}\left( \frac{\langle \bfC\bfx, \bfw \rangle}{\sigma} \right)^2 + \calO(1/\sigma^3) \right]. \]
Let $\mu = \mathbb{E}_{\bfx \sim P}[\bfx]$ denote the expected participation vector. Pushing the expectation through the linear terms gives
\[ \frac{P(\bfy)}{Q(\bfy)} = 1 + \frac{\langle \bfC\mu, \bfw \rangle}{\sigma} + \frac{\mathbb{E}_{\bfx}[\langle \bfC\bfx, \bfw \rangle^2] - \mathbb{E}_{\bfx}[\|\bfC\bfx\|_2^2]}{2\sigma^2} + \calO(1/\sigma^3). \]
To isolate the privacy loss random variable $L(\bfy) = \ln(P(\bfy)/Q(\bfy))$, we apply the logarithmic expansion $\ln(1 + u) = u - u^2/2 + \calO(u^3)$. Substituting, we get 
\[L(\bfy) = \frac{\langle \bfC\mu, \bfw \rangle}{\sigma} + \frac{\mathbb{E}_{\bfx}[\langle \bfC\bfx, \bfw \rangle^2] - \mathbb{E}_{\bfx}[\|\bfC\bfx\|_2^2] - \langle \bfC\mu, \bfw \rangle^2}{2\sigma^2} + \calO(1/\sigma^3). \] 
This asymptotic expression decomposes the privacy loss into a leading $\calO(1/\sigma)$ stochastic term and $\calO(1/\sigma^2)$ second-order terms. As $\sigma \to \infty$, this tail is overwhelmingly dominated by the $\calO(1/\sigma)$ first-order term.

For our comparison, we have matched the expected participation rate $p_0$ between cyclic Poisson and warm-start $b$-min-sep subsampling. For cyclic Poisson with a random initial assignment, an example is assigned a phase uniformly at random from $\{1, \ldots, b\}$ and then sampled with probability $b p_0$ every $b$ iterations. Because the phase is uniform, the marginal probability of participating in any specific iteration $i$ is exactly $\frac{1}{b} \cdot b p_0 = p_0$. For warm-start $b$-min-sep, the initialization explicitly places each example's availability state into its stationary distribution, which guarantees the marginal probability of participation in any iteration is also $p_0$. Therefore, both schemes yield the identical expected participation vector $\mu = p_0 \mathbf{1}$. Consequently, the dominant first-order term $\langle \bfC\mu, \bfw \rangle / \sigma$ is exactly the same for both schemes. While the two methods possess different joint participation distributions that influence the $\calO(1/\sigma^2)$ second-order terms, these higher-order differences are asymptotically dwarfed by the non-zero leading term. Thus, in the limit as $\sigma \to \infty$, the two subsampling schemes converge to identical privacy guarantees.

We remark that the convergence of $b$-min-sep and cyclic Poisson to the same privacy guarantee is not a weakness of our proposed scheme, but rather an indication that both have reached the fundamental limit of the problem. Because $\bfC$ is a $b$-banded Toeplitz matrix, multiplying it by the uniform expected participation vector $\mu = p_0 \mathbf{1}$ yields a vector whose entries are exactly $p_0$ times the sum of the row elements of $\bfC$. For iterations $i \ge b$, this sum is a constant equal to the $\ell_1$ row norm of the bands. While the first $b-1$ entries are smaller due to boundary effects, this difference is insignificant for typical training horizons (for example, $n \approx 1000$ iterations and $b \approx 8$ bands) and can be safely ignored. Without loss of generality, we can assume this row norm is $1$, as any other value acts as a global scaling coefficient that is equivalent to scaling the noise $\sigma$. Under this normalization, the dominant first-order signal for both schemes becomes exactly $p_0 \mathbf{1}$. Crucially, this uniform $p_0 \mathbf{1}$ signal profile matches the normalized signal of full-batch DP-SGD. In full-batch DP-SGD, every example participates in every iteration, but to match the noise scale of a subsampled batch, the aggregate gradient is divided by $1/p_0$ (the ratio of the full dataset size to the expected batch size), effectively scaling the per-example contribution to exactly $p_0 \mathbf{1}$. Since it is known that full-batch DP-SGD achieves the optimal signal-to-noise ratio among all subsampling schemes, by retrieving this $p_0 \mathbf{1}$ profile, $b$-min-sep subsampling achieves the optimal theoretical guarantee available in the $\sigma \to \infty$ regime.

\section{Monte Carlo Accounting for $b$-Min-Sep Subsampling} 
\label{section:montecarlo}

We now give a method for Monte Carlo accounting of $b$-BandMF with $b$-min-sep subsampling. Recall that by Lemma 4.5 of \cite{choquette2023privacy}, it suffices to compute the divergence between $P$ and $Q$ induced by the matrix mechanism $\bfy = \bfC \bfx + \bfz$, where for $P$, $\bfx$'s coordinates are $1$ in iterations where an example is sampled by $b$-min-sep subsampling and $0$ otherwise, and for $Q$, $\bfx$ is the all-zeros vector. For both $P$ and $Q$, $\bfz$ is Gaussian noise. Also recall that for Monte Carlo accounting, the two primitives we need are to be able to sample $\bfy \sim P$, and be able to compute $P(\bfy) / Q(\bfy)$ given a sample $\bfy$. Sampling $\bfy \sim P$ (or $\bfy \sim Q$) is straightforward and efficient\footnote{For arbitrary $\bfC$ and sampling scheme for $\bfx$ it might take $O(n^2)$ time to compute the matrix-vector product $\bfC \bfx$, but because $b$-min-sep subsampling enforces the $b$-min-sep property and $\bfC$ is banded and Toeplitz, it can easily be done in $O(n)$ time.}. The main challenge is computing $P(\bfy) / Q(\bfy)$. To evaluate $P(\bfy)$, we could naively evaluate each term in the sum $P(\bfy) = \sum_\bfx p_\bfx \cdot \calN(\bfC \bfx, \sigma^2 \bfI)(\bfy - \bfC \bfx)$ where $p_\bfx$ denotes the probability of sampling $\bfx$ under $b$-min-sep subsampling, and $\calN(\mu, \sigma^2 \bfI)(\bfz)$ denotes the likelihood of $\bfz$ under $\calN(\mu, \sigma^2 \bfI)$. However, there are exponentially many different values of $\bfx$ that can arise in this sum, hence this method of evaluating $P$ is infeasible in practice (in contrast to e.g. balls-in-bins).

We provide an efficient method that evaluates $P(\bfy)$ using dynamic programming. For $i \in \{1, \ldots, n\}$ where $n$ is the total number of iterations, let $P_i$ and $Q_i$ denote the marginal distributions of $P$ and $Q$ for the remaining iterations $i, \ldots, n$. Let $\mathcal{A}_i$ denote the event that the example of interest has not participated in iterations $i-1, i-2, \ldots, i-b+1$ (so it is available to participate in the $i$th iteration). Define
\[f_i(\bfy) := \frac{P_i(\bfy \mid \mathcal{A}_i)}{Q_i(\bfy \mid \mathcal{A}_i)} = \frac{P_i(\bfy \mid \mathcal{A}_i)}{Q_i(\bfy)}.\]
That is, the likelihood ratio of the remaining process starting at iteration $i$, given that the example has not participated in the previous $b-1$ iterations, where $\bfy = (\bfy_1, \ldots, \bfy_n)$ is a sequence of $n$ outputs of the mechanism. We note that given $\mathcal{A}_i$, the quantities $P_i, Q_i$ are functions of only $\bfy_i, \ldots, \bfy_n$ and are unaffected by $\bfy_1, \ldots, \bfy_{i-1}$ (even though we write the argument as $\bfy$ for convenience), which we will prove. We define a recursion for $f_i$ given by
\begin{equation}\label{eq:mainrecursion}
f_i(\bfy) = (1 - p) \cdot f_{i+1}(\bfy)
+ p \cdot \frac{\mathcal{N}(c_i, \sigma^2 I)}{\mathcal{N}(\mathbf{0}, \sigma^2 I)}(\bfy_i, \ldots, \bfy_{i+(b-1)}) \cdot f_{i+b}(\bfy),
\end{equation}
with the boundary condition $f_i(\bfy) = 1$ for all $i > n$, where $c_i \in \mathbb{R}^b$ is the vector formed by taking the nonzero entries from the $i$th column of $\bfC$ (with right zero-padding if it is one of the last $b-1$ columns).

In the recursion, the first term corresponds to the case in which the example does not participate in iteration $i$, while the second term corresponds to the case in which it does and skips the following $b-1$ iterations before re-entering at $i+b$, contributing a likelihood ratio of $\mathcal{N}(c_i, \sigma^2 I)/\mathcal{N}(\mathbf{0}, \sigma^2 I)(\bfy_i, \ldots, \bfy_{i+(b-1)})$ for the $b$ iterations.

\begin{theorem}\label{thm:dp-correctness}
For a given $b$-banded (and non-negative, lower-triangular) $\bfC$, and any output $\bfy$, we have that \cref{eq:mainrecursion} holds (and thus $f_1(\bfy) = P(\bfy) / Q(\bfy)$) when using \cref{fig:bms_basic} to sample $\bfx$ in $P$.
\end{theorem}

\begin{proof}

We first show that $\bfy_i, \ldots, \bfy_n$ is independent of $\bfy_1, \ldots, \bfy_{i-1}$ given $\mathcal{A}_i$.  Under $Q$, this is immediate since $\bfy \sim \mathcal{N}(\mathbf{0}, \sigma^2 I)$, independent of $\mathcal{A}_i$. Under $P$, by the $b$-banded structure of $\bfC$, the segment $(\bfy_i, \ldots, \bfy_n)$ depends only on $(\bfx_{i-(b-1)}, \ldots, \bfx_n, \bfz_i, \ldots, \bfz_n)$, while $(\bfy_1, \ldots, \bfy_{i-1})$ depends only on $(\bfx_1, \ldots, \bfx_{i-1}, \bfz_1, \ldots, \bfz_{i-1})$. Conditioned on $\mathcal{A}_i$, the overlap $\bfx_{i-(b-1)}, \ldots, \bfx_{i-1}$ is fixed to zeros, and the remaining $\bfx$-segments are independent by construction (as entries separated by at least $b$ indices depend only through participation in the in-between iterations, which is already fixed). Since the $\bfz$'s are i.i.d.\ and independent of $\bfx$, it follows that the two $\bfy$-segments are conditionally independent given $\mathcal{A}_i$.

Next, to show the recursion, under $Q = \mathcal{N}(\mathbf{0}, \sigma^2 I)$, we have
\begin{align}
    Q_i(\bfy) &= \mathcal{N}(0, \sigma^2)(\bfy_i) \cdot Q_{i+1}(\bfy)  \label{eq:first_decomp_of_Q_i} \\
    &= \mathcal{N}(\mathbf{0}, \sigma^2 I)(\bfy_i, \ldots, \bfy_{i+(b-1)}) \cdot Q_{i+b}(\bfy)  \label{eq:second_decomp_of_Q_i}
\end{align}
where if the index goes beyond $t$ we treat the marginal at that portion as the degenerate distribution at the empty set. On the other hand, under $P$ and conditioned on $\mathcal{A}_i$, the example participates with probability $p$, independent of all other events. Decomposing on this Bernoulli choice and using the enforced $b$-min-sep constraint yields
\begin{align*}
    P_i(\bfy \mid \mathcal{A}_i) &= (1-p) \cdot \underbrace{\mathcal{N}(0, \sigma^2)(\bfy_i)}_{\text{no participation at $i$}} \cdot \; \underbrace{P_{i+1}(\bfy \mid \mathcal{A}_{i+1})}_{\text{available at next iteration}} \\
    &\phantom{{}={}} + p \cdot \underbrace{\mathcal{N}(\bfc_i, \sigma^2 I)(\bfy_i, \bfy_{i+1}, \ldots, \bfy_{i+(b-1)})}_{\text{participate at $i$ and forced non-participation for $b-1$ iterations}} \cdot \; \underbrace{P_{i+b}(\bfy \mid \mathcal{A}_{i+b})}_{\text{available $b$ iterations later}}.
\end{align*}
We then divide both sides by $Q_i(\bfy)$. In the first term, the factor $\mathcal{N}(0, \sigma^2)(\bfy_i)$ cancels with the corresponding
factor in \eqref{eq:first_decomp_of_Q_i}, and we recognize the residual ratio as
\[\frac{(1-p) \cdot \mathcal{N}(0, \sigma^2)(\bfy_i) \cdot P_{i+1}(\bfy \mid \mathcal{A}_{i+1})}{\mathcal{N}(0, \sigma^2)(\bfy_i) \cdot Q_{i+1}(\bfy)} = (1-p) \cdot f_{i+1}(\bfy).\]
In the second term, we divide by \eqref{eq:second_decomp_of_Q_i} to get
\[\frac{p \cdot \mathcal{N}(\bfc_i, \sigma^2 I)(\bfy_i, \bfy_{i+1}, \ldots, \bfy_{i+(b-1)}) \cdot P_{i+b}(\bfy \mid \mathcal{A}_{i+b})}{\mathcal{N}(\mathbf{0}, \sigma^2 I)(\bfy_i, \ldots, \bfy_{i+(b-1)}) \cdot Q_{i+b}(\bfy)} = p \cdot \frac{\mathcal{N}(\bfc_i, \sigma^2 I)}{\mathcal{N}(\mathbf{0}, \sigma^2 I)}(\bfy_i, \ldots, \bfy_{i+(b-1)}) \cdot f_{i+b}(\bfy).\]
Combining the two proves the recursive relationship. The boundary condition is $f_i(\bfy)=1$ for all $i>n$ since both $P$ and $Q$ place unit mass on the empty suffix. This proves the stated recursion.

\end{proof}

While we state \cref{thm:dp-correctness} for \cref{fig:bms_basic} for simplicity, as a corollary we can apply it to the warm-start version of $b$-min-sep subsampling in \cref{fig:bms_ws} and show that $\frac{f_1(\bfy) + p \cdot \sum_{i=2}^b f_i(\bfy)}{1 + (b-1)p} = P(\bfy)/Q(\bfy)$.

For Toeplitz $\bfC$, $\{f_1(\bfy), f_2(\bfy), \ldots f_n(\bfy)\}$ can be computed in $O(n \min\{b, \log n\})$ time. Namely, the computation has two main components: (i) computing $\mathcal{N}(\bfc_i, \sigma^2 I) / \mathcal{N}(\mathbf{0}, \sigma^2 I) (\bfy_i, \ldots, \bfy_{i+(b-1)})$ for all $i$, and (ii) evaluating the recursion given these terms. (i) takes $O(t)$ time because to evaluate the likelihood ratio, we need to compute $\|\bfc_i\|^2$ and $\langle \bfc_i, (\bfy_i, \ldots, \bfy_{i+(b-1)}) \rangle$ for all $i$. The former is straightforward to compute in $O(n)$ time because $C$ is Toeplitz, and the latter can be naively done in $O(nb)$ time but can also be computed in $O(n \log n)$ time using FFT, since these are ``sliding'' dot products. (ii) takes $O(n)$ time because we compute $f_i$ for $n$ values of $i$, and each computation takes $O(1)$ time. 

With the efficient computation of $P(\bfy)/Q(\bfy)$, we can apply the Monte Carlo privacy-loss accounting framework described in \cref{section:evr} to obtain $(\varepsilon, \delta)$-DP guarantees (and related primitives such as noise multiplier calibration) for BandMF with $b$-min-sep subsampling.

We have open-sourced an implementation of both sampling from $P$ (or $Q$) and efficiently computing $P/Q$ as part of the \texttt{jax\_privacy} library\footnote{\url{https://github.com/google-deepmind/jax_privacy/tree/main/jax_privacy/experimental/monte_carlo}}.

\section{Multi-Attribution Setting}\label{section:multi-attribution}

We next demonstrate the versatility of $b$-min-sep subsampling compared to other sampling schemes, by showing it can be applied to the \textit{multi-attribution setting}. In the multi-attribution setting, as defined in \cite{ganesh2025s}, rather than just receiving a set of examples, we also receive a set of users $U$ and for each example we are given a subset of users that example is attributed to. As in \cite{ganesh2025s}, we now say that two datasets $D, D' \in \calD^n$ sharing a set of users $U$ are adjacent under the \textit{fixed-graph adjacency} if they are the same except with all examples attributed to one user $u \in U$ replaced with $\bot$, a special element whose gradients are zero everywhere.

Standard ``user-to-example'' reductions, as e.g.\@ are commonly used in the federated learning setting, cannot be applied in the multi-attribution setting because examples are not necessarily attributed to a single user. So to apply DP-SGD to the multi-attribution setting, \cite{ganesh2025s} use a contribution bounding framework where they first preprocess the dataset to limit the number of examples attributed to each user to be at most $k$, then use bounds on the group privacy of DP-SGD (i.e., the adjacency allows any $k$ examples to change between adjacent datasets) with Poisson subsampling (as proven in e.g. \cite{charles25userleveldp}) to calibrate the noise multiplier.  However, due to the lack of group privacy bounds for BandMF with any form of amplification, \cite{ganesh2025s} has to resort to  BandMF without amplification, and so BandMF does not show the strong improvements over DP-SGD that it achieves in the example-level DP setting.

Before showing how $b$-min sep subsampling can resolve this gap, we first explain why neither of the past works on amplification for BandMF can be readily applied to the multi-attribution setting. 

\mypar{Cyclic Poisson} Recall that cyclic Poisson subsampling of \cite{choquette2024amplified} partitions the dataset into $b$ subsets $D_1, \ldots, D_b$, and only samples from each subset every $b$ iterations. To extend their privacy analysis (which effectively can ignore all subsets except the one the sensitive example is in) to the multi-attribution setting we would need to ensure that each user's examples only appear in one of the $b$ subsets $D_i$. In other words, we'd need to partition the users into $b$ subsets $U_1, \ldots, U_b$, and then form data subsets $D_i$ such that the examples in $D_i$ only are attributed to users in $U_i$. We also want the $D_i$ to be roughly balanced in size; e.g. one trivial partition is to put all users in $U_1$, which allows us to place all examples in $D_1$, however this means $D_2, \ldots, D_b$ would be empty and hence we would have empty batches in most iterations. In general satisfying these constraints on the $D_i$ may discard a large portion of examples. For example, consider when for each pair of users, there is one example attributed to that pair. Then with $|U|$ users there are $\approx |U|^2$ examples, but to achieve a balanced partition we would need to have $b$ subsets of $|U|/b$ users, so each data subset only has $\approx |U|^2/b^2$ eligible examples and hence we would only retain fraction $\approx 1 / b$ of the examples. Hence cyclic Poisson subsampling is not broadly applicable to the multi-attribution setting without a large loss in dataset size, even disregarding the algorithmic difficulty of finding a good partition.

\mypar{Balls-in-bins} Unlike cyclic Poisson, balls-in-bins subsampling allows each example to potentially participate in every iteration. One might then hope group privacy properties can be proven for BandMF using balls-in-bins subsampling, which would allow one to re-use the framework of \cite{ganesh2025s}. However, the privacy analysis of balls-in-bins subsampling uses Monte Carlo accounting, where we evaluate $P(\bfy)$ by summing over all participation patterns. For a single example, since it participates in one of $T$ positions in the length-$T$ epoch, the number of participation patterns is just $T$ and hence this sum is tractable. However, to prove group privacy guarantees we would have to consider all joint participation patterns for a set of $k$ examples, i.e. the number of ways to throw $k$ balls into $T$ bins of which there are $\Omega(T^k)$. In other words, the runtime of even computing a single sample grows exponentially in $k$, so balls-in-bins does not permit an efficient group privacy analysis\footnote{We note that while warm-start $b$-min-sep subsampling generalizes balls-in-bins subsampling in the example-level DP setting, as we will discuss later we were unable to exactly extend our warm-start version of $b$-min-sep subsampling to the multi-attribution setting, hence the observation that $b$-min-sep subsampling generalizes balls-in-bins subsampling in the example-level DP setting does not carry over to the multi-attribution setting.}.

\subsection{Multi-attribution $b$-min-sep subsampling}

Now, we present the extension of $b$-min-sep subsampling to the multi-attribution setting as \cref{fig:bms_multiattr}. In this version, in each iteration, every example in the dataset is independently sampled with a fixed probability $p$, and the sampled examples are placed into a temporary holding set for that iteration. Being sampled at this stage does not yet mean that the example will be used in the current iteration. This sampling step is performed regardless of whether the example was selected or excluded in previous iterations. An example is excluded from the current batch if, in any of the previous $b-1$ iterations, there was a sampled example that shares at least one user with it. Importantly, this exclusion rule depends only on whether examples were sampled in those earlier iterations, not on whether they were ultimately included in earlier batches, and we will explain this choice in the analysis. After removing all such excluded examples, the remaining sampled examples form the batch for the current iteration. From the perspective of one user, their examples only participate once in any consecutive window of $b$ iterations, though the one participation potentially includes more than one of their examples.

\begin{figure}[H]
\begin{algorithm}[H]
\caption{Multi-attribution $b$-min-sep subsampling}
\label{fig:bms_multiattr}
\textbf{Parameters:} Sampling probability $p$, minimum separation parameter $b$, dataset $D=\{e_1,e_2,\ldots,e_m\}$, total iterations $n$, user-sets $\{U(e)\}_{e\in D}$ where $U(e)$ denotes the set of users attributed to example $e$.
\begin{algorithmic}[1]
\State \textbf{(User-sharing neighborhood)} For any example $e\in D$, define
\[
N(e) \;=\; \{\, e' \in D \mid U(e') \cap U(e) \neq \emptyset \,\}.
\]
\For{$i \in [n]$}
\State $S_i \leftarrow$ include each element of $D$ independently with probability $p$.
\State $D^{exclude}_i \leftarrow \bigcup_{j=\max\{1, \: i-b+1\}}^{i-1}\ \ \bigcup_{e \in S_j} N(e)$
\State $B_i \leftarrow S_i \setminus D^{exclude}_i$
\EndFor
\State \Return batches $B_1, B_2, \ldots B_n$.
\end{algorithmic}
\end{algorithm}
\end{figure}

Similar to Algorithm~\ref{fig:bms_ws}, we may want to apply a warm start to Algorithm~\ref{fig:bms_multiattr} to ensure a similar batch size for the earlier iterations. The problem is that with the intertwined user-example graph, it is hard to analytically compute the stationary distribution to start with (unlike the prior case where examples participated independently of each other), so instead we can choose to run the subsampling algorithm (but no training) for many iterations, then start training when the batch sizes become stable.

\subsection{Privacy analysis of Algorithm~\ref{fig:bms_multiattr}}

We similarly apply Monte Carlo accounting to analyze the privacy of Algorithm~\ref{fig:bms_multiattr}, which means we need an efficient way to evaluate $P(\bfy) / Q(\bfy)$, where $P$ is the distribution of the output if a user is included and $Q$ if every example belonging to the user is replaced by the zero-gradient $\bot$.

Let $D_u = \{e \in D: u \in U(e)\}$ denote the subset of examples belonging to a particular user $u$ of interest, and $\mathcal{A}_i$ denote the event that every example in $D_u$ is available to participate at the $i$th iteration, i.e.,
\[\mathcal{A}_i \coloneqq \left\{ e \notin D^{exclude}_i, \quad \forall e \in \bigcup_{e' \in D_u} N(e') \right\}.\]

We similarly define $P_i, Q_i$ to be the marginals of $P, Q$ starting from iteration $i$, and
\[f_i(\bfy) \coloneqq \frac{P_i (\bfy \mid \mathcal{A}_i)}{Q_i(\bfy \mid \mathcal{A}_i)} = \frac{P_i (\bfy \mid \mathcal{A}_i)}{Q_i(\bfy)},\]
that is, the likelihood ratio of the remaining process starting from iteration $i$, given that every example belonging to the user is not barred from participation at iteration $i$. Specifically, if we pre-process the dataset such that each user has at most $k_u$ examples just like in \cite{ganesh2025s} (where $k_u$ is a constant), we will prove that $f_1(\bfy)$ is the likelihood ratio of the mechanism's output for the \textit{worst case} dataset, and that $f_1$ is computable by following the recursion
\begin{equation*}
    f_i(\bfy) = (1-p)^{k_u} f_{i+1}(\bfy) + \sum_{j=1}^{k_u} \mathbb{P} \big( \mathrm{Binom}(k_u, p) = j \big) \cdot \frac{\mathcal{N}(j \cdot \bfc_i, \sigma^2 I)}{\mathcal{N}(\mathbf{0}, \sigma^2 I)}(\bfy_i, \ldots, \bfy_{i+(b-1)}) \cdot f_{i+b}(\bfy),
\end{equation*}
with the boundary condition $f_i = 1$ for all $i > n$, where $\bfc_i \in \mathbb{R}^b$ is the vector formed by taking the nonzero entries from the $i$th column of $\bfC$ (again with right zero-padding if it is one of the last $b-1$ columns). The time to compute the dynamic program increases multiplicatively by a factor of $O(k_u)$ to $O(n k_u \min(b, \log n))$ due to the need to compute $k_u$ times as many normal likelihood ratios and sum over them; this is a minor blowup in contrast with the exponential blowup that the accounting for balls-in-bins would incur.

\begin{theorem}
Let $\bfx$ denote the random vector where $\bfx_i = 0$ if any of $\bfx_{i-1}, \bfx_{i-2}, \ldots \bfx_{i-b+1} > 0$, and otherwise $\bfx_i \sim \mathrm{Binom}(k_u, p)$. For a fixed $b$-banded $\bfC$ and $\bfz \sim \calN(\boldzero, \sigma^2\mathbb{I})$ let $P = \bfC \bfx + \bfz$ and $Q = \bfz$. Then:

\begin{enumerate}
    \item for $b$-BandMF on any two adjacent datasets $D, D'$ such that each user has at most $k_u$ examples in $D$, using \cref{fig:bms_multiattr} to form batches satisfies any DP guarantee satisfied by $P, Q$.
    \item $f_1(\bfy) = P(\bfy) / Q(\bfy)$.
\end{enumerate}
\end{theorem}

\begin{proof}
To prove the first part of the theorem, for any dataset $D$ and user $u$, by Lemma 4.5 of \cite{choquette2023privacy} it suffices to consider the case where $P = \bfC \bfx + \bfz$ and $Q = \bfz$ where $\bfx$ is the vector such that $\bfx_i$ is the number of examples attributed to $u$ in the batch $B_i$ formed by running \cref{fig:bms_multiattr} on $D$. Without loss of generality we assume the user has exactly $k_u$ examples (having fewer examples only improves the privacy guarantee). We refer to this $\bfx$ as being induced by $D$ and $u$.

To show it suffices to consider the single-user setting, Lemma 4.3 of \cite{choquette2023privacy} shows that given two randomized vectors $\bfa, \bfb$ and Gaussian $\bfz$, then $(\bfa + \bfz, \bfz)$ satisfy any indistinguishability guarantee satisfied by $(\bfb + \bfz, \bfz)$ if $\bfa, \bfb \geq \boldzero$ and there is a coupling of $\bfa, \bfb$ such that $\bfa \leq \bfb$ w.p. 1. As we assume $\bfC$ is non-negative, for non-negative $\bfx$ and $\bfx'$, $\bfx \leq \bfx'$ implies $\bfC \bfx \leq \bfC \bfx'$. So for any dataset $D$ and user $u$, to reduce to the single-user setting it suffices to show a coupling between $\bfx$ induced by this choice of $D, u$, and $\bfx'$ induced by a different $D', u$ such that $D'$ only has $k_u$ examples all attributed to user $u$, such that $\bfx \leq \bfx'$ w.p. 1.

In particular, for the latter dataset we use the subset $D_u$ of $D$ containing only the $k_u$ examples belonging to $u$. Consider running \cref{fig:bms_multiattr} on $D$ and $D_u$, where we couple the two processes by enforcing that $S_i \cap D_u$ is the same across both runs of \cref{fig:bms_multiattr}. Under this coupling, by definition of $D^{exclude}_i$ we have that $D^{exclude}_i \cap D_u$ when running on $D$ is a superset of $D^{exclude}_i \cap D_u$ when running on $D_u$, hence $B_i \cap D_u = (S_i \setminus D^{exclude}_i) \cap D_u$ when running on $D$ is a subset of $B_i \cap D_u$ when running on $D_u$. This gives $\bfx \leq \bfx'$ for $\bfx$ induced by $D, u$ and $\bfx'$ induced by $D_u, u$ w.p. 1.

To finish proving the first part of the theorem, we need to show that in the single-user setting, forming $\bfx$ by iteratively setting $\bfx_i = 0$ if any of $\bfx_{i-1}, \bfx_{i-2}, \ldots \bfx_{i-b+1} > 0$, and otherwise $\bfx_i \sim \mathrm{Binom}(k_u, p)$ gives a matrix mechanism $\bfC \bfx + \bfz$ dominating the matrix mechanism using $\bfx$ induced by $D_u, u$. Intuitively, this says for $D_u$ if we had used $B_i$ instead of $S_i$ as the sets for exclusion in \cref{fig:bms_multiattr}, the privacy can only get worse. To show this, suppose at least one of the $k_u$ examples in $D_u$ participates in $B_i$ under \cref{fig:bms_multiattr}. Let $i'$ be the smallest value $\geq i$ such that $\calA_{i'}$ holds. $i'$ is a random variable with support $\{i+b, i+b+1, \ldots n\}$. In turn, iterations $i$ to $n$ effectively sample this random variable $i'$, and then recursively run the mechanism for $n - i' + 1$ iterations and pre-pad it with $i' - i + b$ iterations of Gaussian noise. So for any fixed choice of $i' \geq i+b$, the output in iterations $i$ to $n$ is a post-processing of the mechanism of $i' = i+b$ (since we can truncate and pre-pad the output of the mechanism for $i' = i+b$ to simulate these iterations), hence by the convexity property of DP assuming $i' = i + b$ w.p. 1 only makes the privacy analysis worse. Assuming $i' = i + b$ always is equivalent to sampling $\bfx$ as in the theorem statement, so this shows $\bfC \bfx + \bfz$ using $\bfx$ in the theorem statement dominates the mechanism using $\bfx$ induced by $D_u, u$ according to the theorem statement as desired.

The proof of the second part of the theorem is very similar to that of \cref{thm:dp-correctness}. Abusing notation, redefine $\calA_i$ to be the event where $\bfx_{i-1}, \bfx_{i-2}, \ldots \bfx_{i-b+1} = 0$ under the distribution of $\bfx$ given in the theorem statement. Just like before, under $Q = \mathcal{N}(\mathbf{0}, \sigma^2 \bfI)$, we have
\begin{align}
    Q_i(\bfy) &= \mathcal{N}(0, \sigma^2)(\bfy_i) \cdot Q_{i+1}(\bfy)  \label{eq:first_decomp_of_Q_i_multiattr} \\
    &= \mathcal{N}(\mathbf{0}, \sigma^2 \bfI)(\bfy_i, \ldots, \bfy_{i+(b-1)}) \cdot Q_{i+b}(\bfy)  \label{eq:second_decomp_of_Q_i_multiattr}
\end{align}
where if the index goes beyond $t$ we treat the marginal at that portion as the degenerate distribution at the empty set. On the other hand,
\begin{align*}
    P_i(\bfy | \calA_i) &= \underbrace{(1-p)^{k_u}}_{\mathbb{P}(\text{sample no examples})} \cdot \underbrace{\mathcal{N}(\mathbf{0}, \sigma^2)(\bfy_i)}_{\text{$\bfx_i = 0$}} \cdot \; \underbrace{P_{i+1}(\bfy \mid \mathcal{A}_{i+1})}_{\calA_{i+1} \text{ holds}} \\
    &\phantom{{}={}} + \sum_{j=1}^{k_u} \underbrace{\mathbb{P} \big( \mathrm{Binom}(k_u, p) = j \big)}_{\mathbb{P}(\text{$\bfx_i = j$})} \cdot \underbrace{\mathcal{N} ( j\bfc_i, \sigma^2 \bfI) (\bfy_i, \ldots, \bfy_{i+(b-1)})}_{\text{Product of $\bfx_i = j$ and $\bfc_i$}} \cdot \; \underbrace{P_{i+b}(\bfy \mid \mathcal{A}_{i+b})}_{\text{$\bfx_i = 0$ for next $b-1$ iterations}}.
\end{align*}
Divide both sides by $Q_i(\bfy)$, the first term on RHS by \eqref{eq:first_decomp_of_Q_i_multiattr} and the second term by \eqref{eq:second_decomp_of_Q_i_multiattr}, and we recover the proposed recursion.

\end{proof}

In the warm start case, since we instead simulate a warm start, we do not have the starting distribution analytically to do a tight analysis like before. However, the privacy guarantees of the cold start mechanism dominate those of the warm start mechanism (and thus can be used for accounting of the warm start mechanism), since the warm start mechanism is distributionally equivalent to a (randomized) post-processing of the cold start mechanism that truncates the last $\min \{i : \mathcal{A}_i\} - 1$ coordinates.

\subsection{A conjectured better algorithm}

In the Algorithm~\ref{fig:bms_multiattr}, we made the design choice to rule out availability of an example based on the coin toss results of its user-sharing examples (in the previous $b-1$ iterations), instead of based on actual participation. In fact, if we use the latter, the algorithm would still satisfy user-level $b$-min-sep requirement, and it appears that it would lead to a clean analysis too. To put it formally, consider Algorithm~\ref{fig:bms_multiattr_participation}, with a difference only in a subscript when forming $D^{exclude}$.

\begin{figure}[H]
\begin{algorithm}[H]
\caption{Improved multi-attribution $b$-min-sep subsampling}
\label{fig:bms_multiattr_participation}
\textbf{Parameters:} Sampling probability $p$, minimum separation parameter $b$, dataset $D=\{e_1,e_2,\ldots,e_m\}$, total iterations $n$, user-sets $\{U(e)\}_{e\in D}$ where $U(e)$ denotes the set of users attributed to example $e$.
\begin{algorithmic}[1]
\State \textbf{(User-sharing neighborhood)} For any example $e\in D$, define
\[
N(e) \;=\; \{\, e' \in D \mid U(e') \cap U(e) \neq \emptyset \,\}.
\]
\For{$i \in [n]$}
\State $S_i \leftarrow$ include each element of $D$ independently with probability $p$.
\State $D^{exclude}_i \leftarrow \bigcup_{j=\max\{1, \: i-b+1\}}^{i-1}\ \ \bigcup_{e \in \blue{B_j}} N(e)$
\State $B_i \leftarrow S_i \setminus D^{exclude}_i$
\EndFor
\State \Return batches $B_1, B_2, \ldots B_n$.
\end{algorithmic}
\end{algorithm}
\end{figure}

However, our privacy analysis for Algorithm~\ref{fig:bms_multiattr} relies on a key monotonicity property that enables a reduction to a single-user dataset. Namely, after applying the ``aligned-gradients'' reduction of Lemma 4.5 of \cite{choquette2023privacy}, the mean shift in iteration $i$ is proportional to the number of examples from the sensitive user $u$ that participate in that iteration. Under Algorithm~\ref{fig:bms_multiattr}, we coupled the sampling indicators $S_i$ across two runs—one on the full dataset $D$ and one on the restricted dataset $D_u$—so that the interim sampled sets $S_i\cap D_u$ are identical in the two runs. Since Algorithm~\ref{fig:bms_multiattr} forms its exclusions using only the interim sampled sets, any examples outside $D_u$ can only create additional exclusions for $D_u$ (external blocking). Consequently, in every iteration $i$, the number of participating examples from $D_u$ in the full run is at most the number in the restricted run, and hence the associated participation vector is element-wise dominated. In turn, because $\bfC$ is $b$-banded and the participation vectors satisfy the $b$-min-sep property, this yields an element-wise domination of the corresponding Gaussian means in the mixture representation of the mechanism, which is sufficient for the $(\varepsilon,\delta)$-DP domination we use in the proof.

For the participation-based variant (Algorithm~\ref{fig:bms_multiattr_participation}), the same coupling argument breaks down. Although examples outside $D_u$ can still only remove examples from $B_i$ via external blocking, this removal also changes the future exclusion sets because exclusions are now formed from the realized batches $\{B_j\}$ rather than the fixed interim samples $\{S_j\}$. In particular, external blocking can prevent certain examples from appearing in $B_j$, which may reduce future exclusions and thereby allow additional participations of examples in $D_u$ in later iterations. Thus, even under a coupling that shares the interim samples $S_j$, the participation counts for user $u$ need not be element-wise dominated: the full run may have fewer participations early on but more participations later, yielding participation vectors $\bfx,\bfx'$ such that neither $\bfx \ge \bfx'$ nor $\bfx' \ge \bfx$ holds element-wise. Since the aligned-gradients reduction makes the mean shift proportional to these per-iteration participation counts, the resulting Gaussian mean vectors $\bfC \bfx$ and $\bfC \bfx'$ likewise need not be comparable element-wise. Therefore, the sufficient coordinate-wise mean-domination condition (Lemma 4.3 in \cite{choquette2023privacy}) that underpins our reduction proof does not apply to Algorithm~\ref{fig:bms_multiattr_participation}.

Moreover, being unable to prove mean-domination is not merely an artifact of the proof technique. We have identified concrete instances in which, under a coupling of the sampling indicators, the Gaussian mean vector induced by running Algorithm~\ref{fig:bms_multiattr_participation} on the full dataset is not coordinate-wise dominated by the corresponding mean vector obtained by restricting to the single-user subset $D_u$. In these examples, external blocking suppresses early participations but enables later participations that would otherwise be excluded, leading to non-monotone shifts in the mean sequence.

We emphasize that failure of coordinate-wise mean domination does \textbf{not} imply that the
participation-based variant is not dominated in $(\varepsilon,\delta)$-DP; coordinate-wise mean
domination is a sufficient but not necessary condition for such a domination. We conjecture that
the same single-user reduction remains valid for Algorithm~\ref{fig:bms_multiattr_participation},
but we leave establishing this formally as an open problem.

\begin{conjecture}[Validity of Monte Carlo recursion for participation-based blocking]
\label{conj:participation_based_recursion}
Consider the participation-based multi-attribution $b$-min-sep subsampling scheme
(Algorithm~\ref{fig:bms_multiattr_participation}), and fix a user $u$ with at most $k_u$ examples
after preprocessing. Let $P$ and $Q$ denote the output distributions of the banded Gaussian matrix
mechanism under fixed-graph adjacency, where $P$ corresponds to the original dataset and $Q$
corresponds to replacing all examples attributed to user $u$ with $\bot$.

Then the Monte Carlo privacy analysis developed for Algorithm~\ref{fig:bms_multiattr} remains valid
as an upper bound for Algorithm~\ref{fig:bms_multiattr_participation}.
\end{conjecture}

\section{Empirical results}\label{section:experiments}

\subsection{MSE comparisons}

We first compare the MSE of prefix sums achieved by $b$-min-sep sampling to the MSE achieved by cyclic Poisson and balls-in-bins across a grid of parameter settings. The MSE of prefix sums is defined by e.g. \cite{denisov2022improved} as
\[\frac{1}{n} \sum_{i=1}^n \Var{\sum_{j \leq i} (\bfC^{-1}\bfz)_j} = \frac{1}{n}\|\bfA \bfC^{-1}\|_F^2 \sigma(\bfC)^2,\] where $\bfA$ is the all-ones lower triangular matrix, $\bfC$ is the matrix used in DP-MF, and $\sigma(\bfC)$ is the noise necessary for DP under the given $\bfC$ and sampling scheme. Just as reducing the noise multiplier of vanilla DP-SGD is a model/task-agnostic improvement to DP training, \cite{denisov2022improved, choquette2022multi} demonstrate that reducing MSE will generally lead to better training performance, i.e. these comparisons are predictive of the relative performance of the different methods in a wide variety of training settings.

We fix a number of iterations $n = 1024$, and $\delta = 10^{-3}$. Since the runtime and memory of BandMF increases with the number of bands \cite{choquette2024amplified, mckenna2025scaling} and hence large numbers of bands are typically impractical, we consider a compute-constrained setting where we restrict to at most 32 bands used during training. We vary (1) the number of bands of $\bfC$ in $\{2, 4, 8, 16, 32\}$, for each using $\bfC$ that minimizes MSE without amplification as in \cite{choquette2024amplified} (2) the ``epoch length,'' i.e. ratio of dataset size to average batch size $|D| / B$ in $\{128, 256, 512, 1024\}$, (3) $\epsilon$ in $\{0.5, 1, 2, 4, 8, 16\}$. For $b$-min-sep subsampling, we additionally vary the min-sep parameter $b$ in $\{2, 4, 8, \ldots, 1024\}$, restricted to $b \leq |D| / B$ and $b$ is at least the number of bands in $\bfC$. In Monte Carlo accounting, we estimate $\epsilon$ at $\delta / 2$, and use \cref{fig:evr_nm} in \cref{section:evr} to calibrate the noise multiplier.

In \cref{fig:msevscyclicconstrained} and \cref{fig:msevsbibconstrained} we plot the ratios of MSEs of achieved by $b$-min-sep subsampling and the other methods (less than 1 means an improvement), using the best choice of bands for each method in each setting. Since $|D|/B$-min-sep subsampling retrieves balls-in-bins, the ratios in \cref{fig:msevsbibconstrained} are never greater than 1.  In this setting $b$-min-sep offers meaningful improvements over cyclic Poisson except when both $\epsilon$ and epoch length are low. For balls-in-bins we see that $b$-min-sep always improves on balls-in-bins except when the epoch length is 1024, i.e. we only use 1 epoch. Since balls-in-bins reuses a fixed participation pattern across epochs, the $i$th iteration of any epoch is highly correlated with the $i$th iteration of any other epoch, which greatly weakens privacy amplification, making its privacy amplification weaker in the multi-epoch setting, but not in the single-epoch setting.

In short, our MSE comparisons suggest that in single-epoch settings one can default to balls-in-bins as the best method, and when both $\epsilon$ and $|D|/B$ are small cyclic Poisson may be preferable, but otherwise $b$-min-sep is the best-performing method.

\begin{figure}
    \centering
    \begin{minipage}{0.45\textwidth}
        \centering
        \includegraphics[width=\textwidth]{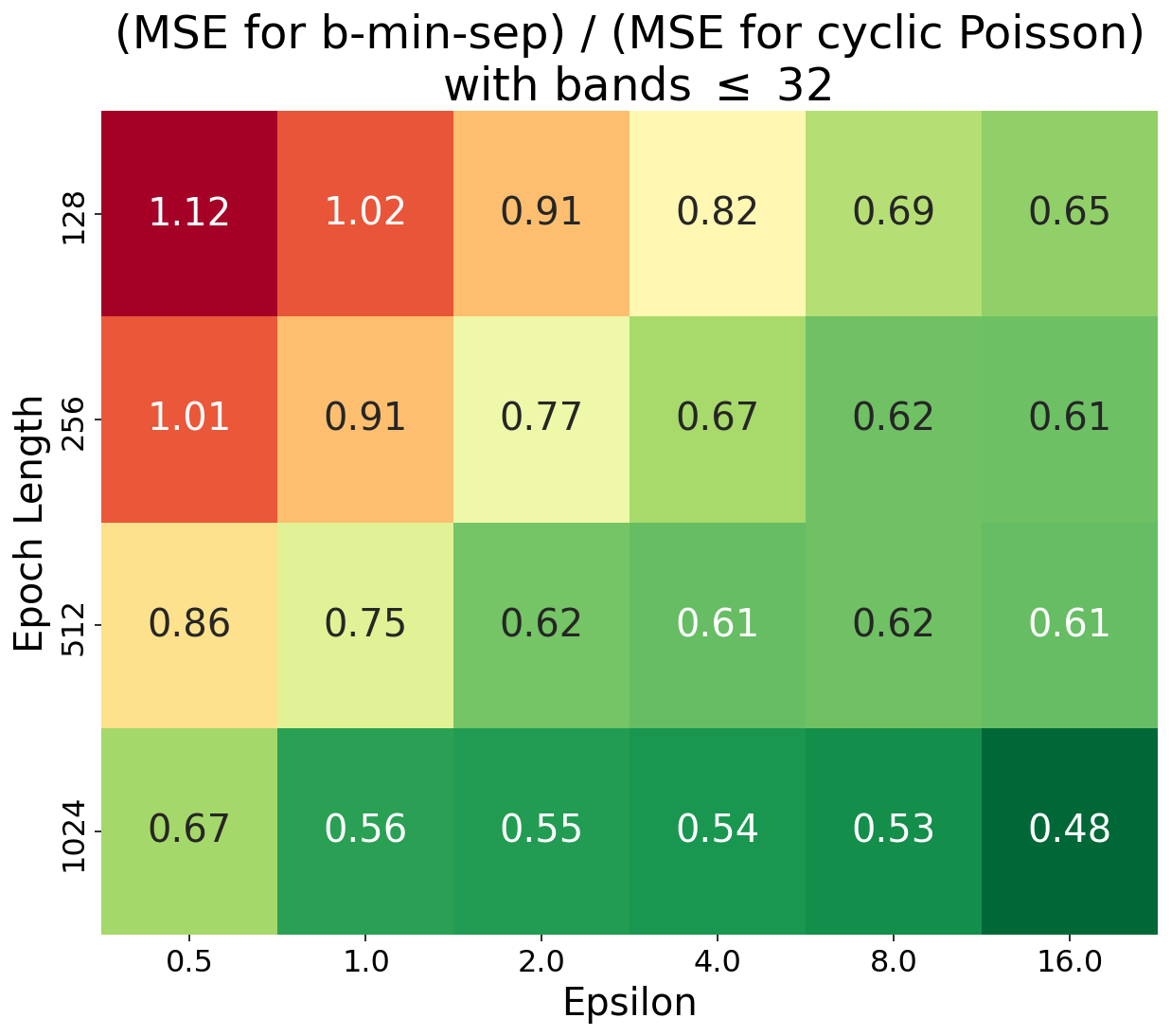} 
        \caption{Comparison to cyclic Poisson.}
    \label{fig:msevscyclicconstrained}
    \end{minipage}\hfill
    \begin{minipage}{0.45\textwidth}
        \centering
        \includegraphics[width=\textwidth]{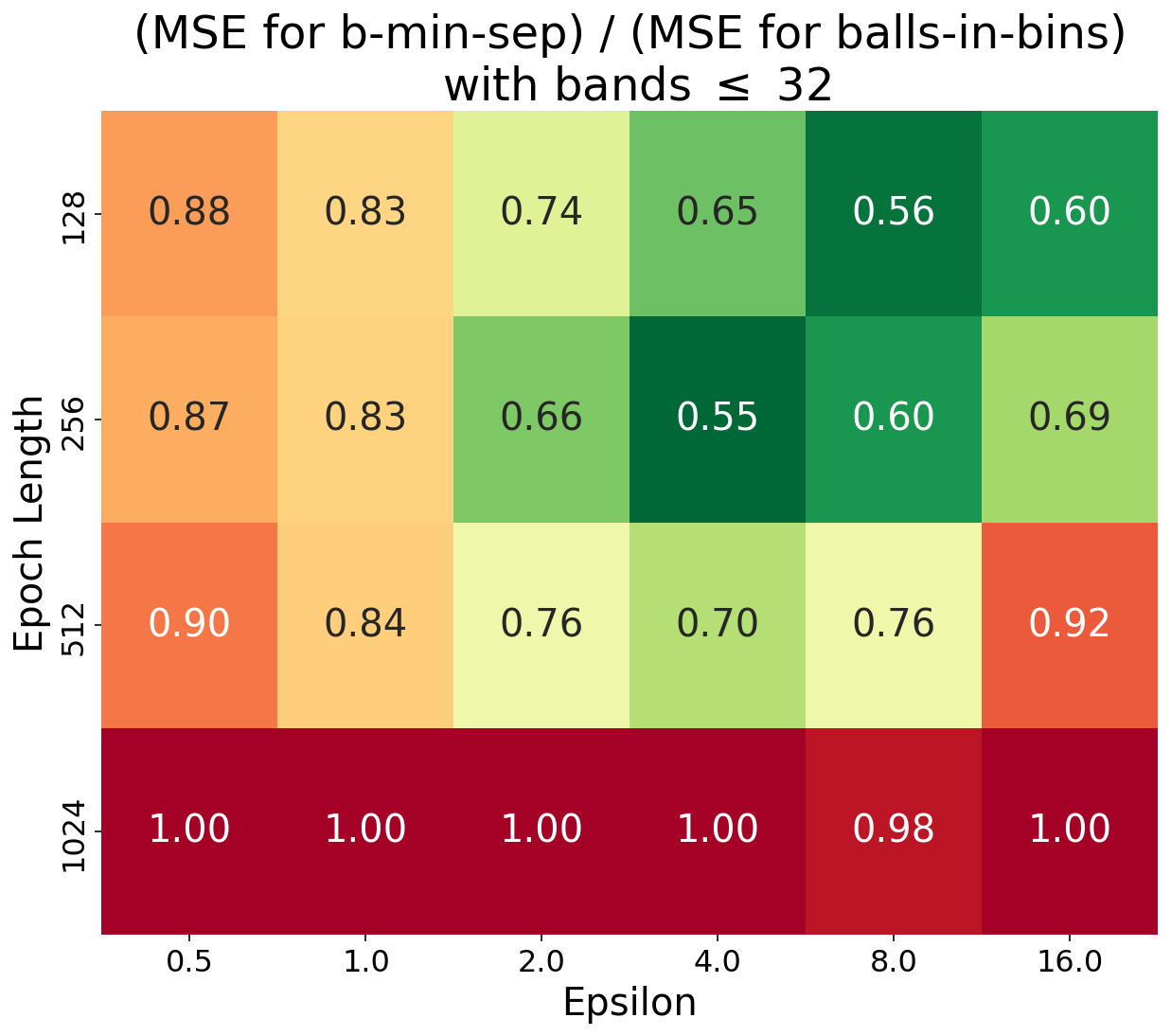} 
        \caption{Comparison to balls-in-bins.}
    \label{fig:msevsbibconstrained}
    \end{minipage}
\end{figure}

\subsection{Experiments on CINIC10}

\begin{figure}[t]
    \centering
    \includegraphics[width=\linewidth]{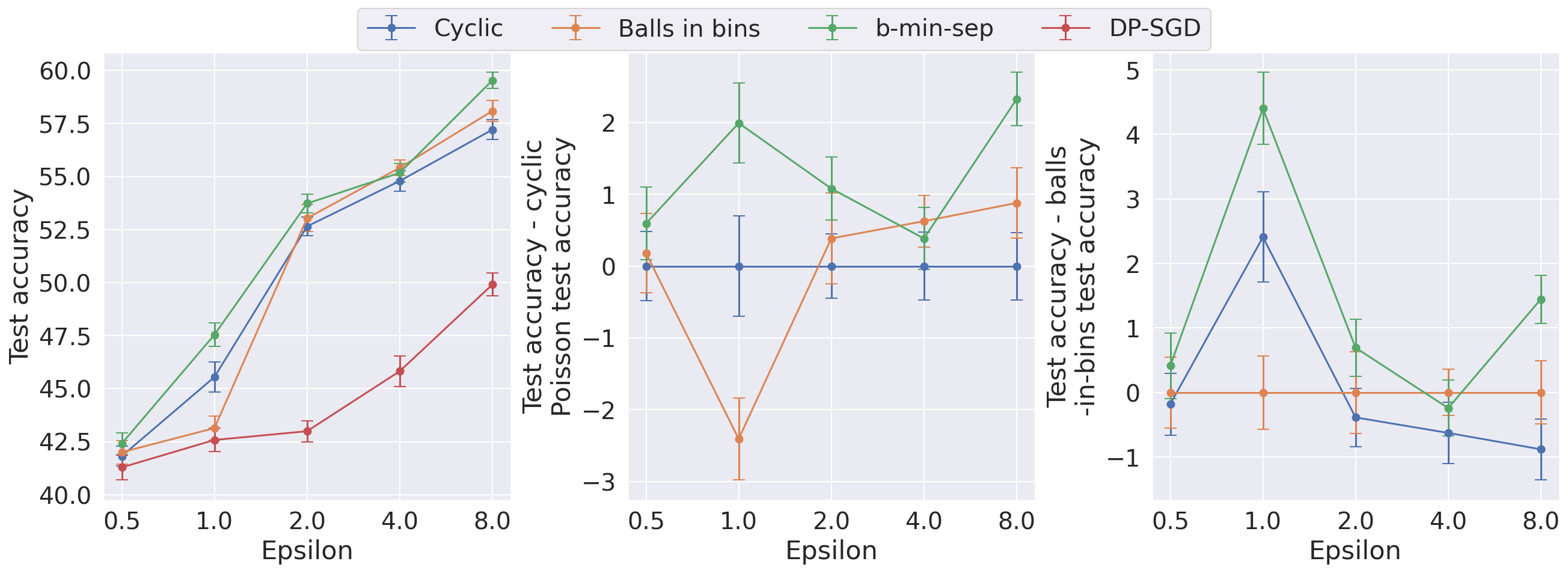}
    \caption{Left: Average test accuracy of a VGG trained on CINIC10 with 30 trials and 95\% confidence intervals for different $\varepsilon$ values and sampling schemes. Middle/right: The difference between test accuracy of each sampling scheme to cyclic Poisson/balls-in-bins.}
    \label{fig:cinic_acc}
\end{figure}

\begin{figure}[t]
    \centering
    \includegraphics[width=0.7\linewidth]{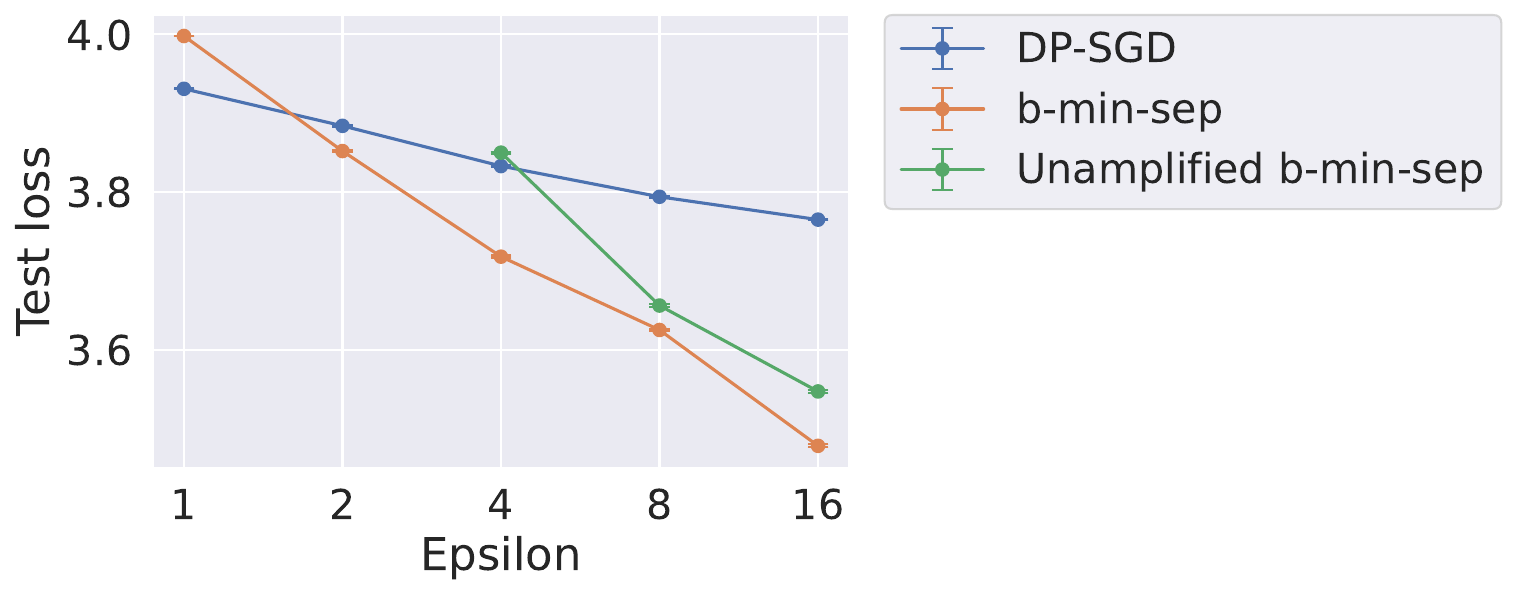}
    \caption{Average test loss on arXiv fine-tuning task with 30 trials and 95\% confidence intervals for different $\varepsilon$ values and sampling schemes.}
    \label{fig:arxiv_loss}
\end{figure}

\mypar{Experimental setup} CINIC10 \cite{darlow18cinic}, is an image classification task formed by extending CIFAR10 using images from ImageNet. We train a VGG  model \cite{simonyan2014very} on the CINIC10 dataset for 1500 iterations with expected batch size $B$ such that $|D|/B$ = 300.  We set the expected batch size to be $|D|/300$, and train for 1500 iterations (i.e., each example appears 5 times in expectation). Each training run is on a standard VGG architecture, and was done using only a single CPU for about 6 hours. We use $\delta = \frac{1}{2|D|}$. Our underlying optimizer is SGD with momentum of $0.95$ and a linear learning rate decay from the base learning rate $\eta$ to $\eta / 20$ over iterations $n/4$ to $3n/4$. We sweep over base learning rate in $\{0.01, 0.03, 0.1, 0.3, 1.0\}$.

We compare DP-SGD, and $(\leq 64)$-BandMF using cyclic Poisson subsampling, balls-in-bins subsampling, and $b$-min-sep subsampling. For each sampling scheme for BandMF, we swept the number of bands in $\{2, 4, 8, 16, 32, 64\}$. While we could have swept larger numbers of bands, we again remark that more bands requires more computation, hence restricting to 64 bands is a reasonable representation of a typical compute-constrained setting in practice. For each number of bands, we use the choice of $\bfC$ that minimizes the MSE without assuming amplification, under a fixed number of participations and a (\# bands)-min-sep constraint (see \cite{choquette2024amplified} for more details on this procedure). For $b$-min-sep, we additionally sweep the min-sep parameter $b$ in $\{$(\# bands)$, 150, 300\}$.   We sweep $\epsilon \in \{0.5, 1, 2, 4, 8\}$. 

\mypar{Results} In \cref{fig:cinic_acc} we plot test accuracy with 95\% confidence intervals. For an easier comparison between different sampling schemes for BandMF, we also replot the test accuracies, after subtracting the test accuracy of cyclic Poisson/balls-in-bins, dropping DP-SGD for ease of readability.

\mypar{Discussion} In almost all settings $b$-min-sep is the best method, except for $\epsilon = 4$ where it is (statistically insignificantly) slightly worse than balls-in-bins, but otherwise it is often noticeably better than the other sampling schemes, improving on cyclic Poisson by as much as $2\%$ absolute accuracy and balls-in-bins by as much as $4\%$ absolute accuracy. Note that our MSE comparisons predicted that $b$-min-sep would be the best method as long as the ratio $|D|/B$ was sufficiently large, and we are training for multiple epochs, and both conditions are met here.

\subsection{Multi-attribution experiments on arXiv}

\mypar{Experimental setup} Following \cite{ganesh2025s}, we fine-tune a TinyBERT model \cite{devlin2018pretraining} on the arXiv dataset \cite{clement2019arxiv}.

The arXiv dataset consists of $\approx$ 1.9 million papers posted on arXiv. For each paper, we treat its abstract as the contents of the example and the authors of the paper as the users the example is attributed to. To make the graph structure more challenging / reflective of the multi-attribution setting, we employ a preprocessing step not employed by \cite{ganesh2025s} where we remove any paper only authored by a single author; i.e. all examples are attributed to multiple users. We shuffle the dataset and take the last $200,000$ examples as a test set. We then run the contribution bounding algorithm of \cite{ganesh2025s} for DP-SGD on the training set, to arrive at a preprocessed version of the dataset where at most $k_u = 2$ examples are attributed to each user, perhaps with duplicates (we choose $2$ and to allow duplicates as this was the best choice of the contribution bounding hyperparameters in \cite{ganesh2025s}). The size of the resulting training set is approximately $300,000$; we fix $\delta = 1 / 500,000$. We fix $1000$ iterations of training with an expected batch size of $1024$, using Adam with ``scale-then-privatize'' as discussed in \cite{ganesh2025design}, and sweep the learning rate over the range $[3 \cdot 10^{-7}, 3 \cdot 10^{-2}]$. Each training run was done on a single V100 GPU for about 6 hours.

To achieve a ``warm-start'' we form $2000$ batches using $b$-min-sep subsampling where all examples are available for sampling in the first iteration, and only use the last $1000$ batches in training. It suffices to compute the privacy guarantees as if we had formed $1000$ batches without a warm-start. Since it is challenging to exactly compute the expected batch size after the warm-start as a function of the sampling probability for $b$-min-sep, we instead use binary search over the sampling probability to find one that empirically achieves approximately the expected batch size. We found that the empirical average batch size achieved by $b$-min-sep subsampling was between $1023$ and $1025$ for all the sampling probabilities we arrived at, hence any error in the calculation of the sampling probability due to empirical variance should only negligibly affect our results. We sweep the number of bands $b$ in $\{2, 4, 8, 16, 32\}$. We sweep $\varepsilon \in \{1, 2, 4, 8, 16\}$. We compare DP-SGD, warm-start DP-BandMF with $b$-min-sep subsampling, and unamplified DP-BandMF.

For unamplified DP-BandMF, we run the batch formation algorithm of \cite{ganesh2025s} which forms a series of batches satisfying $b$-min-sep and $k$-max-part, i.e. no user participates more than $k$ times across all batches, sufficient conditions to account for unamplified BandMF as a Gaussian mechanism. For unamplified BandMF, we use $128$-BandMF and choose $b = 128, k = 7$, as $128$-min-sep is the largest power of $2$ such that the contribution bounding algorithm succeeded, and simultaneously $7$ was the smallest value of $k$ across all $b$ such that the algorithm succeeded. As the utility decreases with $k$ and increases with $b$, this results in the highest-utility variant of unamplified BandMF.

\mypar{Results} In \cref{fig:arxiv_loss} we plot test loss with 95\% confidence intervals. We do not report results for unamplified BandMF at $\varepsilon \leq 2$ because training diverged even with a learning rate as small as $3 \cdot 10^{-7}$ for Adam, whereas reasonable learning rates in $[3 \cdot 10^{-4}, 10^{-2}]$ were stable for non-divergent runs.

\mypar{Discussion} While DP-SGD with Poisson subsampling performs best at $\varepsilon=1$, for $\varepsilon \geq 2$ $b$-min-sep subsampling with BandMF outperforms both DP-SGD and unamplified BandMF. Note that $1$-min-sep subsampling with $1$-BandMF is equivalent to DP-SGD with Poisson subsampling, so the weaker performance of $b$-min-sep at smaller $\varepsilon$ arises only because we restrict to $b>1$. Thus, $b$-min-sep is always as good as, and often better than, both DP-SGD and unamplified DP-MF.

\subsection{Production Use-Case}

We next use $b$-min-sep sampling to improve the privacy-utility curve of an existing production application of DP-BandMF for private training, to demonstrate both its gains in practical regimes and also that Monte Carlo accounting is feasible in practice. Due to the confidential nature of the application, we do not discuss details of this application besides training hyperparameters here\footnote{Even discussing downstream metrics such as training loss may impact the privacy of the training data or the confidentality of the use-case.}. The accounting details are model/task-agnostic and remain verifiable/reproducible without these details.

The existing implementation\footnote{In actuality these numbers vary depending on the results of the data generation process, but these numbers are representative of a typical use case.} uses 256-BandMF (where $\bfC$ optimizes the MSE under cyclic Poisson sampling) combined with a truncated and padded variant of cyclic Poisson sampling (necessary for gradient compilations which require fixed batch sizes) defined in \cite{chua2024scalable} to fine-tune a model for 7200 iterations, using an expected batch size of 1793 and truncated batch size of 2048, with dataset size of 14,745,600, and is trained using $(10, \approx 1.301 \cdot 10^{-8})$-DP. This requires a noise multiplier of $\sigma \approx 0.627$. We first define a truncated and padded variant of $b$-min-sep sampling and give its privacy analysis in \cref{sec:truncated}, as well as deriving efficient Monte Carlo accouting primitives for this variant. Then, we ran Monte Carlo accounting for $256$-min-sep sampling with $7.5 \cdot 10^{9}$ samples per verification for $\sigma = 0.56, 0.5, 0.47, 0.44$, keeping all other training parameters the same and using  \cref{fig:evr_nm} in \cref{section:evr} to turn this into a formal end-to-end training guarantee. The samples were generated in 300 CPU-days per noise multiplier (parallelized across 150 CPUs) using our open-sourced libraries. While this is a significant amount of computation, the internal costs of this computation are cheap, especially compared to other parts of this production pipeline such as data preparation and training itself which required running on costly accelerators.

Our accounting validates that a noise multiplier of $0.47$ for $b$-min-sep achieves the same end-to-end $(\epsilon, \delta$)-DP guarantee, a $25\%$ reduction in the noise multiplier compared to cyclic Poisson. Since all other training parameters are the same as the original training setup and a difference in batch formation while maintaining the same expected batch size does not meaningfully affect utility compared to the effect of the DP noise, the $25\%$ reduction in noise maps to a significant improvement in utility of an application of DP-MF at no cost in privacy budget. Furthermore, compared to the theoretical (not-tight) lower bound of $\sigma \approx 0.368$ achieved by DP-SGD with vanilla Poisson sampling, $b$-min-sep closes $60\%$ of the gap to this lower bound. 

\section*{Acknowledgements}

We are thankful to

\begin{itemize}
    \item Brendan McMahan for the initial suggestion that $b$-min-sep subsampling combined with BandMF might be amenable to a privacy analysis.
    \item Nikita Kalinin and Jan Schuchardt for discussions on whether our techniques could be extended to other classes of matrices and whether matrices which do not optimize the MSE might be better for learning.
    \item Ayfer Özgür for helpful comments on an earlier draft of the paper.
    \item Ryan McKenna for careful reviews of the open-sourced version of our code.
    \item Anonymous reviewers for their suggestions on improving the paper.
\end{itemize}

\bibliographystyle{alpha}
\bibliography{ref}

@misc{pillutla2025correlatednoisemechanismsdifferentially,
      title={Correlated Noise Mechanisms for Differentially Private Learning}, 
      author={Krishna Pillutla and Jalaj Upadhyay and Christopher A. Choquette-Choo and Krishnamurthy Dvijotham and Arun Ganesh and Monika Henzinger and Jonathan Katz and Ryan McKenna and H. Brendan McMahan and Keith Rush and Thomas Steinke and Abhradeep Thakurta},
      year={2025},
      eprint={2506.08201},
      archivePrefix={arXiv},
      primaryClass={cs.LG},
      url={https://arxiv.org/abs/2506.08201}, 
}

@article{darlow18cinic,
  author       = {Luke Nicholas Darlow and
                  Elliot J. Crowley and
                  Antreas Antoniou and
                  Amos J. Storkey},
  title        = {{CINIC-10} is not ImageNet or {CIFAR-10}},
  journal      = {CoRR},
  volume       = {abs/1810.03505},
  year         = {2018},
  url          = {http://arxiv.org/abs/1810.03505},
  eprinttype   = {arXiv},
  eprint       = {1810.03505},
  bibsource    = {dblp computer science bibliography, https://dblp.org}
}

@misc{clement2019arxiv,
    title={On the Use of ArXiv as a Dataset},
    author={Colin B. Clement and Matthew Bierbaum and Kevin P. O'Keeffe and Alexander A. Alemi},
    year={2019},
    eprint={1905.00075},
    archivePrefix={arXiv},
    primaryClass={cs.IR}
}

@inproceedings{henzinger2023almost,
  title={Almost tight error bounds on differentially private continual counting},
  author={Henzinger, Monika and Upadhyay, Jalaj and Upadhyay, Sarvagya},
  booktitle={Proceedings of the 2023 Annual ACM-SIAM Symposium on Discrete Algorithms (SODA)},
  pages={5003--5039},
  year={2023},
  organization={SIAM},
  url={https://arxiv.org/abs/2211.05006},
}

@inproceedings{henzinger2024unifying,
  title={A unifying framework for differentially private sums under continual observation},
  author={Henzinger, Monika and Upadhyay, Jalaj and Upadhyay, Sarvagya},
  booktitle={Proceedings of the 2024 Annual ACM-SIAM Symposium on Discrete Algorithms (SODA)},
  pages={995--1018},
  year={2024},
  organization={SIAM}
}

@article{denisov2022improved,
  title={Improved differential privacy for sgd via optimal private linear operators on adaptive streams},
  author={Denisov, Sergey and McMahan, H Brendan and Rush, John and Smith, Adam and Guha Thakurta, Abhradeep},
  journal={Advances in Neural Information Processing Systems},
  volume={35},
  pages={5910--5924},
  year={2022}
}

@misc{gu2025correlatingcrossiterationnoisedpsgd,
      title={Correlating Cross-Iteration Noise for DP-SGD using Model Curvature}, 
      author={Xin Gu and Yingtai Xiao and Guanlin He and Jiamu Bai and Daniel Kifer and Kiwan Maeng},
      year={2025},
      eprint={2510.05416},
      archivePrefix={arXiv},
      primaryClass={cs.LG},
      url={https://arxiv.org/abs/2510.05416}, 
}

@inproceedings{wang2023randomized,
author = {Wang, Jiachen T. and Mahloujifar, Saeed and Wu, Tong and Jia, Ruoxi and Mittal, Prateek},
title = {A randomized approach to tight privacy accounting},
year = {2023},
publisher = {Curran Associates Inc.},
address = {Red Hook, NY, USA},
booktitle = {Proceedings of the 37th International Conference on Neural Information Processing Systems},
articleno = {1469},
numpages = {38},
location = {New Orleans, LA, USA},
series = {NIPS '23}
}

@article{doroshenko2022connect,
  author       = {Vadym Doroshenko and
                  Badih Ghazi and
                  Pritish Kamath and
                  Ravi Kumar and
                  Pasin Manurangsi},
  title        = {Connect the Dots: Tighter Discrete Approximations of Privacy Loss
                  Distributions},
  journal      = {Proc. Priv. Enhancing Technol.},
  volume       = {2022},
  number       = {4},
  pages        = {552--570},
  year         = {2022},
  url          = {https://doi.org/10.56553/popets-2022-0122},
  doi          = {10.56553/POPETS-2022-0122},
  bibsource    = {dblp computer science bibliography, https://dblp.org}
}

@INPROCEEDINGS{kasiviswanathan2008what,
  author={Kasiviswanathan, Shiva Prasad and Lee, Homin K. and Nissim, Kobbi and Raskhodnikova, Sofya and Smith, Adam},
  booktitle={2008 49th Annual IEEE Symposium on Foundations of Computer Science}, 
  title={What Can We Learn Privately?}, 
  year={2008},
  volume={},
  number={},
  pages={531-540},
  doi={10.1109/FOCS.2008.27}}

@InProceedings{chua25ballsandbins,
  title = 	 {Balls-and-Bins Sampling for DP-SGD},
  author =       {Chua, Lynn and Ghazi, Badih and Harrison, Charlie and Kamath, Pritish and Kumar, Ravi and Leeman, Ethan Jacob and Manurangsi, Pasin and Sinha, Amer and Zhang, Chiyuan},
  booktitle = 	 {Proceedings of The 28th International Conference on Artificial Intelligence and Statistics},
  pages = 	 {946--954},
  year = 	 {2025},
  editor = 	 {Li, Yingzhen and Mandt, Stephan and Agrawal, Shipra and Khan, Emtiyaz},
  volume = 	 {258},
  series = 	 {Proceedings of Machine Learning Research},
  month = 	 {03--05 May},
  publisher =    {PMLR},
  url = 	 {https://proceedings.mlr.press/v258/chua25a.html}
}

@inproceedings{balle2018improving,
  title={Improving the Gaussian Mechanism for Differential Privacy: Analytical Calibration and Optimal Denoising},
  author={Balle, Borja and Wang, Yu-Xiang},
  booktitle={International Conference on Machine Learning},
  pages={394--403},
  year={2018}
}

@article{simonyan2014very,
  title={Very deep convolutional networks for large-scale image recognition},
  author={Simonyan, Karen and Zisserman, Andrew},
  journal={arXiv preprint arXiv:1409.1556},
  year={2014}
}

@inproceedings{song2013stochastic,
  title={Stochastic gradient descent with differentially private updates},
  author={Song, Shuang and Chaudhuri, Kamalika and Sarwate, Anand D},
  booktitle={2013 IEEE Global Conference on Signal and Information Processing},
  pages={245--248},
  year={2013},
  organization={IEEE}
}

@article{choquette2024amplified,
  title={(Amplified) Banded Matrix Factorization: A unified approach to private training},
  author={Choquette-Choo, Christopher A and Ganesh, Arun and McKenna, Ryan and McMahan, H Brendan and Rush, John and Guha Thakurta, Abhradeep and Xu, Zheng},
  journal={Advances in Neural Information Processing Systems},
  volume={36},
  year={2024}
}

@INPROCEEDINGS {dvijotham2024efficient,
author = { Dvijotham, Krishnamurthy Dj and McMahan, H. Brendan and Pillutla, Krishna and Steinke, Thomas and Thakurta, Abhradeep },
booktitle = { 2024 IEEE 65th Annual Symposium on Foundations of Computer Science (FOCS) },
title = {{ Efficient and Near-Optimal Noise Generation for Streaming Differential Privacy }},
year = {2024},
volume = {},
ISSN = {},
pages = {2306-2317},
doi = {10.1109/FOCS61266.2024.00135},
url = {https://doi.ieeecomputersociety.org/10.1109/FOCS61266.2024.00135},
publisher = {IEEE Computer Society},
address = {Los Alamitos, CA, USA},
month =Oct}

@inproceedings{kalinin2024banded,
  title={Banded Square Root Matrix Factorization for Differentially Private Model Training},
  author={Nikita P. Kalinin and Christoph H. Lampert},
  booktitle={Proceedings of the 38th International Conference on Neural Information Processing Systems (NeurIPS 2024)},
  volume={37},
  year={2024},
  url={arxiv.org}
}

@inproceedings{
mckenna2025scaling,
title={Scaling up the Banded Matrix Factorization Mechanism for Large Scale Differentially Private {ML}},
author={Ryan McKenna},
booktitle={The Thirteenth International Conference on Learning Representations},
year={2025},
url={https://openreview.net/forum?id=69Fp4dcmJN}
}

@misc{kalinin2025squarerootsoptimalbound,
      title={Back to Square Roots: An Optimal Bound on the Matrix Factorization Error for Multi-Epoch Differentially Private SGD}, 
      author={Nikita P. Kalinin and Ryan McKenna and Jalaj Upadhyay and Christoph H. Lampert},
      year={2025},
      eprint={2505.12128},
      archivePrefix={arXiv},
      primaryClass={cs.CR},
      url={https://arxiv.org/abs/2505.12128}, 
}

@INPROCEEDINGS{charles25userleveldp,
  author={Charles, Zachary and Ganesh, Arun and McKenna, Ryan and McMahan, H. Brendan and Mitchell, Nicole and Pillutla, Krishna and Rush, Keith},
  booktitle={2025 IEEE Conference on Secure and Trustworthy Machine Learning (SaTML)}, 
  title={Learning with User-Level Differential Privacy Under Fixed Compute Budgets}, 
  year={2025},
  volume={},
  number={},
  pages={901-920},
  doi={10.1109/SaTML64287.2025.00055}}

@inbook{henzinger25improved,
author = {Monika Henzinger and Jalaj Upadhyay},
title = {Improved Differentially Private Continual Observation Using Group Algebra},
booktitle = {Proceedings of the 2025 Annual ACM-SIAM Symposium on Discrete Algorithms (SODA)},
chapter = {},
pages = {2951-2970},
doi = {10.1137/1.9781611978322.95},
year={2025},
URL = {https://epubs.siam.org/doi/abs/10.1137/1.9781611978322.95},
eprint = {https://epubs.siam.org/doi/pdf/10.1137/1.9781611978322.95}
}

@inproceedings{DMNS,
    author  = {Dwork, Cynthia and McSherry, Frank and Nissim, Kobbi and Smith, Adam},
    title   = {Calibrating Noise to Sensitivity in Private Data Analysis},
    booktitle = {Proc.  of the Third Conf.  on Theory of Cryptography (TCC)},
    year    = {2006},
    pages   = {265--284},
    url     = {http://dx.doi.org/10.1007/11681878\_14},
}

@inproceedings{
chua2024scalable,
title={Scalable {DP}-{SGD}: Shuffling vs. Poisson Subsampling},
author={Lynn Chua and Badih Ghazi and Pritish Kamath and Ravi Kumar and Pasin Manurangsi and Amer Sinha and Chiyuan Zhang},
booktitle={The Thirty-eighth Annual Conference on Neural Information Processing Systems},
year={2024},
url={https://openreview.net/forum?id=6gMnj9oc6d}
}

@article{choquette2023correlated,
  title={Correlated noise provably beats independent noise for differentially private learning},
  author={Choquette-Choo, Christopher A and Dvijotham, Krishnamurthy and Pillutla, Krishna and Ganesh, Arun and Steinke, Thomas and Thakurta, Abhradeep},
  journal={arXiv preprint arXiv:2310.06771},
  year={2023}
}

@article{choquette2022multi,
  title={Multi-epoch matrix factorization mechanisms for private machine learning},
  author={Choquette-Choo, Christopher A and McMahan, H Brendan and Rush, Keith and Thakurta, Abhradeep},
  journal={arXiv preprint arXiv:2211.06530},
  year={2022}
}

@inproceedings{BST14,
 author = {Bassily, Raef and Smith, Adam and Thakurta, Abhradeep},
 title = {Private Empirical Risk Minimization: Efficient Algorithms and Tight Error Bounds},
 booktitle = {Proc.  of the 2014 IEEE 55th Annual Symp.  on Foundations of Computer Science (FOCS)},
 year = {2014},
 pages = {464--473},
}

@inproceedings{DP-DL,
  author    = {Mart{\'{\i}}n Abadi and
               Andy Chu and
               Ian J. Goodfellow and
               H. Brendan McMahan and
               Ilya Mironov and
               Kunal Talwar and
               Li Zhang},
  title     = {Deep Learning with Differential Privacy},
  booktitle = {Proc.  of the 2016 {ACM} {SIGSAC} Conf.  on Computer and Communications Security ({CCS}'16)},
  pages     = {308--318},
  year      = {2016},
}

@misc{ganesh2025design,
      title={On Design Principles for Private Adaptive Optimizers}, 
      author={Arun Ganesh and Brendan McMahan and Abhradeep Thakurta},
      year={2025},
      eprint={2507.01129},
      archivePrefix={arXiv},
      primaryClass={cs.LG},
      url={https://arxiv.org/abs/2507.01129}, 
}

@misc{ganesh2025tighter,
      title={Tighter Privacy Analysis for Truncated Poisson Sampling}, 
      author={Arun Ganesh},
      year={2025},
      eprint={2508.15089},
      archivePrefix={arXiv},
      primaryClass={cs.CR},
      url={https://arxiv.org/abs/2508.15089}, 
}

@misc{devlin2018pretraining,
  author = {Devlin, Jacob and Chang, Ming-Wei and Lee, Kenton and Toutanova, Kristina},
  description = {[1810.04805] BERT: Pre-training of Deep Bidirectional Transformers for Language Understanding},
  note = {cite arxiv:1810.04805Comment: 13 pages},
  title = {BERT: Pre-training of Deep Bidirectional Transformers for Language
  Understanding},
  url = {http://arxiv.org/abs/1810.04805},
  year = 2018
}

@article{ponomareva23dpfy,
   title={How to DP-fy ML: A Practical Guide to Machine Learning with Differential Privacy},
   volume={77},
   ISSN={1076-9757},
   url={http://dx.doi.org/10.1613/jair.1.14649},
   DOI={10.1613/jair.1.14649},
   journal={Journal of Artificial Intelligence Research},
   publisher={AI Access Foundation},
   author={Ponomareva, Natalia and Hazimeh, Hussein and Kurakin, Alex and Xu, Zheng and Denison, Carson and McMahan, H. Brendan and Vassilvitskii, Sergei and Chien, Steve and Thakurta, Abhradeep Guha},
   year={2023},
   month=jul, pages={1113–1201} }

@article{choquette2023privacy,
  title={Privacy amplification for matrix mechanisms},
  author={Choquette-Choo, Christopher A and Ganesh, Arun and Steinke, Thomas and Thakurta, Abhradeep},
  journal={arXiv preprint arXiv:2310.15526},
  year={2023}
}

@article{choquette2024near,
  title={Near exact privacy amplification for matrix mechanisms},
  author={Choquette-Choo, Christopher A and Ganesh, Arun and Haque, Saminul and Steinke, Thomas and Thakurta, Abhradeep},
  journal={arXiv preprint arXiv:2410.06266},
  year={2024}
}

@article{demko1984decay,
  title={Decay rates for inverses of band matrices},
  author={Demko, Stephen and Moss, William F and Smith, Philip W},
  journal={Mathematics of computation},
  volume={43},
  number={168},
  pages={491--499},
  year={1984}
}

@article{ganesh2025s,
  title={It's My Data Too: Private ML for Datasets with Multi-User Training Examples},
  author={Ganesh, Arun and McKenna, Ryan and McMahan, Brendan and Smith, Adam and Wu, Fan},
  journal={arXiv preprint arXiv:2503.03622},
  year={2025}
}

@article{dong2025leveraging,
  title={Leveraging randomness in model and data partitioning for privacy amplification},
  author={Dong, Andy and Chen, Wei-Ning and Ozgur, Ayfer},
  journal={Proceedings of the 42nd International Conference on Machine Learning},
  volume={267},
  pages={13938--13962},
  year={2025}
}

@article{chen2023privacy,
  title={Privacy amplification via compression: Achieving the optimal privacy-accuracy-communication trade-off in distributed mean estimation},
  author={Chen, Wei-Ning and Song, Dan and Ozgur, Ayfer and Kairouz, Peter},
  journal={Advances in Neural Information Processing Systems},
  volume={36},
  pages={69202--69227},
  year={2023}
}

\appendix

\crefalias{section}{appendix}

\section{Estimate-Verify-Release Without Empty Outputs}\label{section:evr}

For convenience we propose \cref{fig:evr_nm}, a variant of the estimate-verify-release framework for choosing the noise multiplier. Unlike the framework given in \cite{wang2023randomized}, under mild assumptions (which we later discuss how to satisfy for the mechanisms of interest in this paper) our variant has the property that it does not ever return an empty output. Furthermore, it allows one to propose many candidate mechanisms and (roughly speaking) run the lowest-noise one that succeeds verification, whereas the framework in \cite{wang2023randomized} requires one to commit to a single mechanism ahead of time. 

\begin{figure}[H]
\begin{algorithm}[H]
\caption{Estimate-Verify-Release with Monte Carlo accounting for noise selection}
\label{fig:evr_nm}
\textbf{Parameters:} Candidate mechanisms $\calM_1, \calM_2, \calM_3, \ldots, \calM_k$ such that $\calM_i$ satisfies any $(\varepsilon, \delta)$-DP guarantee satisfied by $\calM_{i-1}$ for $i \leq k-1$, verifier $\calV$ with outputs $0, 1$.
\begin{algorithmic}[1]
\State{$v_k \leftarrow 1$ }
\For{$i \in [k-1]$}
\State $v_i \leftarrow \calV(\calM_{i})$
\EndFor
\State $i^* = \min\{i: \forall j \geq i, v_j = 1\}$
\State Run $\calM_{i^*}$ and return its output.
\end{algorithmic}
\end{algorithm}
\end{figure}

\begin{theorem}\label{thm:evr}
In addition to the ordering on the mechanisms, assume the following hold:
\begin{itemize}
    \item For any $i$ such that $\calM_i$ does not satisfy $(\varepsilon, \delta)$-DP, $\calV$ independently ``rejects'' $\calM_i$ w.p. at least $1 - q$, i.e. we have $\Pr[v_i = 1] < q$ for all such $i$ and the $v_i$ use independent randomness\footnote{i.e. we use independent samples in each instance of the Monte Carlo estimation procedure}.
    \item $\calM_k$ satisfies $(\varepsilon, \delta)$-DP.
\end{itemize}
Then \cref{fig:evr_nm} satisfies $(\varepsilon, \delta + q(1-\delta))$-DP.
\end{theorem}
\begin{proof}
The proof is almost identical to the proof of Theorem 9 in \cite{wang2023randomized}. Let $o \in \calO$ be the output of $\calM_{i^*}$ in \cref{fig:evr_nm}. By the post-processing property, we can assume the output of \cref{fig:evr_nm} is $(i^*, o)$. Let $i_{\min}$ be the smallest $i$ such that $\calM_i$ (and hence $\calM_j$ for all $j \geq i$) satisfies $(\varepsilon, \delta)$-DP. If $i_{\min} = 1$ then $i^* \geq i_{\min}$ w.p. 1. Otherwise, by the assumption on $\calV$, with probability at least $1-q$ we have that $v_{i_{\min} - 1} = 0$ and thus $i^* \geq i_{\min}$. Then letting $\calM$ denote \cref{fig:evr_nm}, we have for any $S \subseteq [k] \times \calO$ and any two adjacent datasets $D$:

\begin{align*}
&\phantom{{}={}} \Pr_{(i^*, o) \sim \calM(D)}[(i^*, o) \in S]\\
&= \Pr_{(i^*, o) \sim \calM(D)}[(i^*, o) \in S \land i^* \geq i_{\min}] + \Pr_{(i^*, o) \sim \calM(D)}[(i^*, o) \in S \land i^* < i_{\min}] \\
&\leq \Pr_{(i^*, o) \sim \calM(D)}[(i^*, o) \in S \land i^* \geq i_{\min}] + \Pr_{(i^*, o) \sim \calM(D)}[i^* < i_{\min}] \\
&= \sum_{i = i_{\min}}^k \Pr_{o \sim \calM_i(D)}[(i, o) \in S] \cdot \Pr_{(i^*, o) \sim \calM(D)}[i^* = i] + \Pr_{(i^*, o) \sim \calM(D)}[i^* < i_{\min}] \\
&\leq \sum_{i = i_{\min}}^k \left(e^\varepsilon \Pr_{o \sim \calM_i(D')}[(i, o) \in S] + \delta\right) \cdot \Pr_{(i^*, o) \sim \calM(D')}[i^* = i] + \Pr_{(i^*, o) \sim \calM(D')}[i^* < i_{\min}] \\
&= e^\varepsilon \Pr_{(i^*, o) \sim \calM(D')}[(i^*, o) \in S \land i^* \geq i_{\min}]  + \delta \cdot \Pr_{(i^*, o) \sim \calM(D')}[i^* \geq i_{\min}] + \Pr_{(i^*, o) \sim \calM(D')}[i^* < i_{\min}] \\
&\leq e^\varepsilon \Pr_{(i^*, o) \sim \calM(D')}[(i^*, o) \in S \land i^* \geq i_{\min}]  + \delta \cdot \Pr_{(i^*, o) \sim \calM(D')}[i^* \geq i_{\min}] + \Pr_{(i^*, o) \sim \calM(D')}[i^* < i_{\min}] \\
&\leq e^\varepsilon \Pr_{(i^*, o) \sim \calM(D')}[(i^*, o) \in S]  + \delta \cdot \Pr_{(i^*, o) \sim \calM(D')}[i^* \geq i_{\min}] + \Pr_{(i^*, o) \sim \calM(D')}[i^* < i_{\min}] \\
&= e^\varepsilon \Pr_{(i^*, o) \sim \calM(D')}[(i^*, o) \in S]  + \delta + \Pr_{(i^*, o) \sim \calM(D')}[i^* < i_{\min}] \cdot (1-\delta)\\
&\leq e^\varepsilon \Pr_{(i^*, o) \sim \calM(D')}[(i^*, o) \in S]  + \delta + q \cdot (1-\delta).\\
\end{align*}
\end{proof}

To e.g. apply \cref{thm:evr} to noise multiplier selection, we choose a sequence of increasing noise multipliers $\sigma_1, \ldots, \sigma_{k-1}$ and for $i \leq k-1$ the mechanism $\calM_i$ is defined as DP-MF using $\sigma_i$. Increasing the noise multiplier can only improve privacy, hence this choice of $\{\calM_i\}$ satisfies the ordering property assumed by \cref{fig:evr_nm}. For our experiments, we choose our sequence to be a set of powers of $1.01$. We choose the minimum noise multiplier $\sigma_1$ to be smaller than the noise multiplier required for DP-SGD with Poisson sampling (i.e, DP-MF with $\bfC = \bfI$). As DP-MF benefits less from privacy amplification and hence requires more noise than DP-SGD, this is a conservative cutoff, and indeed we never end up choosing $\sigma_1$ in our experiments. To ensure that $\calM_k$ satisfies the target DP guarantee, we can simply have $\calM_k$ be a mechanism that is easy to do privacy accounting for, e.g. DP-SGD with Poisson subsampling or DP-BandMF with cyclic Poisson subsampling.

For the verifier, we use Monte Carlo accounting, which we recall is given a pair of distributions $P, Q$ whose divergence we want to calculate, and draws $s$ samples $y_1, y_2, \ldots, y_s$ from $P$ and computes $\tilde{\delta} = \frac{1}{s} \sum_{i=1}^s \max\left\{1-e^\varepsilon \frac{Q(y_i)}{P(y_i)}\right\}$. For any $\delta' < \delta$ and the given number of samples $s$, Fact 4.2 of \cite{chua25ballsandbins} gives an easily calculable bound $q(\delta, \delta', s)$ on the probability that $\tilde{\delta} \leq \delta'$ if the true hockey-stick divergence between $P, Q$ is larger than $\delta$, hence we can use the verifier outputting 1 if $\tilde{\delta} \leq \delta'$ and 0 otherwise in \cref{fig:evr_nm}, and the resulting bound $q$ in \cref{thm:evr}. Our overall privacy parameter can be written as the function $\delta_{o}(\delta, \delta', s) := \delta + q(\delta, \delta', s)(1-\delta)$.

For our experiments, given target final delta $\delta_t$, for simplicity we fix $\delta' = \delta_t/2$. We use SciPy's scalar minimization library to find $\delta_{o, \min}(\delta', s) := \min_\delta \delta_{o}(\delta, \delta', s)$ for our choice of $\delta'$ and a given $s$, and then use binary search to find the minimum $s$ such that $\delta_{o, \min}(\delta', s) \leq \delta_t$. Our final $(\varepsilon, \delta_t)$-DP mechanism is then to use \cref{fig:evr_nm} with Monte Carlo accounting as the verifier, $\delta' = \delta_t/2$, and this choice of $s$. One could optimize the sample complexity further by optimizing  $\delta_o$ jointly over $\delta', \delta$ instead of just over $\delta$, although this might result in using smaller $\delta'$ and hence a larger ``cost'' for Monte Carlo accounting. 

We have open-sourced code for \cref{fig:evr_nm} as well as computing one of $\delta_t, \delta', \delta_s$ given the other two as part of \texttt{jax\_privacy}\footnote{\url{https://github.com/google-deepmind/jax_privacy/blob/main/jax_privacy/experimental/monte_carlo/delta_calculation.py}}.

\section{Extension to Non-Banded Matrices}

The dynamic programming recursion in \cref{section:montecarlo} relied on a key structural property: conditioned on the availability event $\mathcal{A}_i$, the past outputs $(\bfy_1,\ldots,\bfy_{i-1})$ and the future outputs $(\bfy_i,\ldots,\bfy_t)$ are conditionally independent under both $P$ and $Q$. This follows from locality in the representation $\bfy=\bfC\bfx+\bfz$ when $\bfC$ is $b$-banded: the coordinate $\bfy_i$ only depends on $\bfx_{i-b+1},\ldots,\bfx_i$. Consequently, the likelihood ratio $P/Q$ factorizes over time and yields an exact recursion.

In this appendix we consider the complementary regime where $\bfC^{-1}$ (rather than $\bfC$) is $b$-banded. We assume $\bfC$ is lower triangular and entry-wise nonnegative, but we no longer assume $\bfC$ is banded. In this setting, the conditional independence used by the exact recursion breaks down: even under $\mathcal{A}_i$, the coordinate $\bfy_i$ may depend (through $\bfC$) on all earlier participations, so no finite window fully isolates the influence of early iterations---each participation has a decaying but persistent effect on all future coordinates.

To handle this, we give an analytical ``middle ground'' between (i) a loose upper bound that drops all cross terms in the dense setting, and (ii) the exact recursion available when $\bfC$ is banded. The idea is to enforce a configurable separation length $r$ between participations of the same example (an $r$-min-sep constraint), where $r$ is a hyperparameter. As an example where this could be effective, when $\bfC^{-1}$ is banded, columns of $\bfC$ decay exponentially away from the diagonal, so the influence of a participation decays exponentially into the future. By preventing re-participation for $r-1$ iterations, we eliminate the region where this influence is strongest, while treating the residual long-range influence via a controlled approximation (dropping nonnegative cross terms). This $r$-min-sep formulation interpolates between the exact factorization of the $b$-banded $\bfC$ case (recovered when $r=b$) and a coarse constant-dropping bound in the fully dense case (recovered when $r=1$). We emphasize that the method applies to any lower-triangular, entry-wise nonnegative $\bfC$, including the $b$-banded case.

\subsection{$r$-min-sep recursion for arbitrary (non-banded) $\bfC$}\label{app:rminsep}

We use the same setup as in \cref{section:montecarlo}, except that we enforce $r$-min-sep subsampling (instead of $b$-min-sep) and do not assume $\bfC$ is banded. Let $\mathcal{A}_i$ denote the event that the example of interest is available at iteration $i$ under the $r$-min-sep rule, i.e.\ it has not participated in the previous $r-1$ iterations. Define
\[
f_i(\bfy) \;:=\; \frac{P_i(\bfy \mid \mathcal{A}_i)}{Q_i(\bfy)} ,
\]
where $P_i,Q_i$ are the suffix marginals on iterations $i,\ldots,t$ as in \cref{section:montecarlo}. We upper bound $f_i$ via the recursion
\begin{equation}\label{eq:rminsep_recursion}
f_i(\bfy)
= (1-p)\cdot f_{i+1}(\bfy)
+ p \cdot \frac{\mathcal{N}(\bfc_i,\sigma^2 I)}{\mathcal{N}(\mathbf{0},\sigma^2 I)}(\bfy)\cdot f_{i+r}(\bfy),
\end{equation}
with boundary condition $f_i(\bfy)=1$ for all $i>t$, where $c_i$ denotes the $i$th column of $\bfC$ (viewed as a length-$t$ vector, with entries above the diagonal equal to $0$).

\begin{lemma}\label{lem:rminsep_upperbound}
For all $i\in[t]$ and all $\bfy$, the function defined by \eqref{eq:rminsep_recursion} satisfies
\[
f_i(\bfy) \;\ge\; \frac{P_i(\bfy \mid \mathcal{A}_i)}{Q_i(\bfy)}.
\]
In particular, $f_1(\bfy)$ is an upper bound on the likelihood ratio $P(\bfy)/Q(\bfy)$ and can be used in Monte Carlo accounting.
\end{lemma}

We first quickly justify why an upper bound on $P/Q$ can be used in Monte Carlo accounting. Recall that Monte Carlo accounting estimates $\delta$ as the average of $\mathbb{E}_{y \sim P} [\max\{1 - e^\varepsilon \frac{Q(y)}{P(y)}, 0\}]$. The function $\max\{1 - e^\varepsilon \frac{Q(y)}{P(y)}, 0\}$ is increasing in $P(y)/Q(y)$, i.e. using an upper bound on $P(y)/Q(y)$ can only result in overestimating $\delta$, i.e. a conservative estimate of the privacy parameters.

\begin{proof}

Let $\mu_r(\bfx)=\mu_r(\bfx\mid \mathcal{A}_1)$ denote the distribution of the participation vector $\bfx\in\{0,1\}^t$ under $r$-min-sep subsampling. Let $\bfC_{i:t}$ be the submatrix formed by rows $i$ through $t$ and columns $i$ through $t$ (technically, shorthand for $\bfC_{i:t,\, i:t}$). Using the standard Gaussian likelihood-ratio identity, we can write
\begin{align*}
\frac{P}{Q}(\bfy)
&= \sum_{\bfx\in\{0,1\}^t} \mu_r(\bfx)\,
\exp\!\left(\frac{1}{2\sigma^2}\Big(2\bfy^\top \bfC\bfx - \|\bfC\bfx\|_2^2\Big)\right) \quad\quad (\star)\\
&= \sum_{\bfx_{2:t}\in\{0,1\}^{t-1}} \mu_r\!\left(\begin{bmatrix}0\\ \bfx_{2:t}\end{bmatrix}\right)
\exp\!\left(\frac{1}{2\sigma^2}\Big(2\bfy^\top \bfC \begin{bmatrix}0\\ \bfx_{2:t}\end{bmatrix}
- \Big\|\bfC\begin{bmatrix}0\\ \bfx_{2:t}\end{bmatrix}\Big\|_2^2\Big)\right) \\
&\phantom{{}={}}+
\sum_{\bfx_{2:t}\in\{0,1\}^{t-1}} \mu_r\!\left(\begin{bmatrix}1\\ \bfx_{2:t}\end{bmatrix}\right)
\exp\!\left(\frac{1}{2\sigma^2}\Big(2\bfy^\top \bfC \begin{bmatrix}1\\ \bfx_{2:t}\end{bmatrix}
- \Big\|\bfC\begin{bmatrix}1\\ \bfx_{2:t}\end{bmatrix}\Big\|_2^2\Big)\right) \\
&= \underbrace{\mu_r(\bfx_1=0)}_{=\,1-p}\,
\sum_{\bfx_{2:t}} \mu_r(\bfx_{2:t}\mid \bfx_1=0)
\exp\!\left(\frac{1}{2\sigma^2}\Big(2\bfy_{2:t}^\top \bfC_{2:t}\bfx_{2:t}
- \|\bfC_{2:t}\bfx_{2:t}\|_2^2\Big)\right) \\
&\phantom{{}={}}+
\underbrace{\mu_r(\bfx_1=1)}_{=\,p}\,
\sum_{\bfx_{(1+r):t}} \mu_r(\bfx_{(1+r):t}\mid \bfx_1=1)
\exp\!\left(\frac{1}{2\sigma^2}\Big(2\bfy^\top \bfC \begin{bmatrix}e_1\\ \bfx_{(1+r):t}\end{bmatrix}
- \Big\|\bfC\begin{bmatrix}e_1\\ \bfx_{(1+r):t}\end{bmatrix}\Big\|_2^2\Big)\right),
\end{align*}
where $e_1=(1,0,\ldots,0)^\top\in\mathbb{R}^r$. In the first term we used that $\bfC$ is lower triangular: once we zero out the first coordinate of $\bfx$, the first output coordinate depends only on that zeroed coordinate and hence separates, leaving the suffix $(2{:}t)$. In the second term, $r$-min-sep implies $\bfx_2=\cdots=\bfx_r=0$ whenever $\bfx_1=1$, so we can skip directly to index $1+r$.

We upper bound the exponent in the second term by dropping nonnegative cross terms, using that $\bfC$ is entry-wise nonnegative.We revisit the problem of privacy amplification for BandMF, i.e. DP-SGD with random batches and noise correlated across iterations using a banded correlation matrix. We propose and analyze a new sampling scheme for BandMF called $b$-min-sep subsampling. We show that $b$-min-sep subsampling is a generalization and unification of Poisson subsampling and balls-in-bins subsampling, and that it theoretically improves on the privacy-utility tradeoff of cyclic Poisson sampling, thus subsuming all past practical sampling schemes for DP-SGD and BandMF in the literature. We show that privacy analysis for BandMF with $b$-min-sep subsampling can be performed using Monte Carlo accounting, along with our main technical tool, a dynamic program for computing a likelihood ratio between an exponentially large mixture of Gaussians and the zero-centered Gaussian. Our experiments show that (when using a moderate sample complexity for Monte Carlo accounting), $b$-min-sep subsampling gives empirical improvements over past approaches for medium-to-large $\varepsilon$ values, and that if the error in Monte Carlo estimation could be removed, $b$-min-sep subsampling would Pareto dominate all past approaches for all tested values of $\varepsilon$, validating our theoretical observations. We additionally demonstrate that unlike other methods for amplification of BandMF, $b$-min-sep subsampling naturally extends to the multi-attribution setting. 
First,
\[
\bfy^\top \bfC \begin{bmatrix}e_1\\ \bfx_{(1+r):t}\end{bmatrix}
= \bfy_1^\top \bfc_1 + \bfy_{(1+r):t}^\top \bfC_{(1+r):t}\bfx_{(1+r):t}.
\]
Second,
\begin{align*}
\left\|\bfC \begin{bmatrix}e_1\\ \bfx_{(1+r):t}\end{bmatrix}\right\|_2^2
&= \|\bfc_1\|_2^2 + \sum_{i=1+r}^t \bfx_i \|c_i\|_2^2
+ \underbrace{2\sum_{j>r} \bfx_j\, \bfc_1^\top c_j}_{\ge 0}
+ 2\sum_{r<i<j} \bfx_i\bfx_j\, c_i^\top c_j \\
&\ge \|\bfc_1\|_2^2 + \sum_{i=1+r}^t \bfx_i \|c_i\|_2^2
+ 2\sum_{r<i<j} \bfx_i\bfx_j\, c_i^\top c_j \\
&= \|\bfc_1\|_2^2 + \|\bfC_{(1+r):t}\bfx_{(1+r):t}\|_2^2.
\end{align*}
Plugging these bounds into $(\star)$ yields
\begin{align*}
&\phantom{{}={}} \frac{P}{Q}(\bfy) \\
&\le (1-p)\sum_{\bfx_{2:t}} \mu_r(\bfx_{2:t}\mid \mathcal{A}_2)\,
\exp\!\left(\frac{1}{2\sigma^2}\Big(2\bfy_{2:t}^\top \bfC_{2:t}\bfx_{2:t}
- \|\bfC_{2:t}\bfx_{2:t}\|_2^2\Big)\right) \\
&\phantom{{}={}}+
p \cdot \exp\!\left(\frac{\bfy_1^\top \bfc_1 + \|\bfc_1\|_2^2}{2\sigma^2}\right)
\sum_{\bfx_{(1+r):t}} \mu_r(\bfx_{(1+r):t}\mid \mathcal{A}_{1+r})\,
\exp\!\left(\frac{2\bfy_{(1+r):t}^\top \bfC_{(1+r):t}\bfx_{(1+r):t}
- \|\bfC_{(1+r):t}\bfx_{(1+r):t}\|_2^2}{2\sigma^2} \right).
\end{align*}
The two sums are the same form as $(\star)$ but on shorter suffixes (equivalently, after shifting indices). Therefore, applying the same argument inductively on suffix length yields exactly the recursion \eqref{eq:rminsep_recursion}, proving that the resulting $f_i(\bfy)$ upper bounds the true likelihood ratio on the suffix:
\[
f_i(\bfy)\;\ge\; \frac{P_i(\bfy\mid \mathcal{A}_i)}{Q_i(\bfy)}.\qedhere
\]
\end{proof}

\subsection{Tightness when $\bfC^{-1}$ is banded Toeplitz}\label{app:rminsep_tightness}

We specialize to the common BandMF setting where $\bfC^{-1}$ is a $b$-banded Toeplitz noise-cancellation matrix. Writing the $b$ band values as $h_1,\ldots,h_b$, define the polynomial
\[
H(\lambda)=h_1+h_2\lambda+\cdots+h_b\lambda^{b-1}.
\]
The noise-cancellation condition assumes $H$ has all zeros strictly inside the unit disk. Equivalently, the convolution operator $z\mapsto y$ defined by $y=h*z$ (where $*$ denotes discrete convolution) is invertible on $\ell^2$, and its inverse is bounded: $\|z\|_2\le M\|y\|_2$ for some $M<\infty$, implying $\sigma_{\min}(\bfC^{-1})\ge 1/M$. On the other hand, $\sigma_{\max}(\bfC^{-1})$ is bounded since both the bandwidth and the band values are bounded. Thus $\bfC^{-1}$ is well-conditioned. By the Demko--Moss--Smith theorem \cite{demko1984decay}, there exist constants $K>0$ and $\rho\in(0,1)$ depending only on $h_1,\ldots,h_b$ such that
\[
|\bfC_{i,j}|\le K\rho^{|i-j|}\qquad \forall i,j.
\]
In particular, entries in each column of $\bfC$ (below the diagonal) decay exponentially at rate $\rho$.

The slack in \cref{lem:rminsep_upperbound} arises from dropping terms of the form $\bfc_i^\top \bfc_j$ when $\bfx_i=\bfx_j=1$, which under $r$-min-sep can only happen when $|i-j|\ge r$. Because $\bfC$ is lower triangular Toeplitz (the class of lower triangular Toeplitz matrices is closed under inversion), columns $\bfc_i$ and $\bfc_j$ are shifted versions of each other. For $|i-j|\ge r$, we therefore have
\[
\bfc_i^\top \bfc_j
\;\le\; K^2\left(\rho^{r}+\rho^{r+2}+\rho^{r+4}+\cdots\right)
\;=\; \frac{K^2\rho^r}{1-\rho^2},
\]
which decays exponentially in $r$. Plugging this into the exponent shows that the approximation error vanishes exponentially fast in $r$: equivalently, the recursion yields a multiplicative slack of $1+O(e^{-cr})$ (or an additive slack $O(e^{-cr})$) for some $c>0$ depending only on $\rho$.

\subsection{Practical Considerations and Empirical Observations}

We have implemented and experimented with the $r$-min-sep subsampling strategy described above in settings where $\bfC^{-1}$ is banded, including the Toeplitz noise-cancellation matrices commonly used in BandMF mechanisms. While the analysis in \cref{app:rminsep,app:rminsep_tightness} shows that bandedness of $\bfC^{-1}$ leads to exponentially decaying long-range interactions in $\bfC$, our empirical findings suggest that this structure does not translate into a noticeable improvement in utility compared to the simpler case where $\bfC$ itself is banded.

In particular, for comparable privacy parameters and sampling rates, we did not observe a clear utility advantage from using banded $\bfC^{-1}$ over the banded-$\bfC$ constructions analyzed in the main body, but find it preferable to work directly with banded $\bfC$, where the resulting analysis is both cleaner and tighter. For this reason, we focus on banded-$\bfC$ mechanisms in the main body and include the banded-$\bfC^{-1}$ case here primarily for completeness.

That said, the framework developed here is general and does not rely on any specific structure beyond lower triangularity and entry-wise nonnegativity of $\bfC$. Should future work identify more effective ways to exploit the structure of $\bfC^{-1}$—for instance, via tighter control of cross terms or alternative conditioning arguments—the $r$-min-sep analysis presented here would apply directly.

We hope that including this discussion may nonetheless be useful to readers interested in alternative matrix factorizations or noise-cancellation mechanisms, and that it may help spark further ideas within the community.

\section{An Additional Parameter in the Contribution Bounding Algorithm}\label{sec:additionalparam}

In the pre-processing step of the dataset, we prune examples such that each user has at most $k_u$ examples. This follows the lead of \cite{ganesh2025s} and is a crucial part to the privacy analysis of Algorithm~\ref{fig:bms_multiattr} too. However, we do not set any limitations on the size of an example's user-sharing neighborhood, $|N(e)|$. Intuitively, a large $|N(e)|$ means that including this example $e$ may bar many other examples from participating, leading to a smaller batch size and decreased utility. However, setting a limit $k_e$ on $|N(e)|$ means pruning more examples in the pre-processing step, which also means a smaller batch size, so we need to determine the tradeoff and the parameter $k_e$.

To quantify the effect of $k_e$ on the batch size, let us consider a toy example, where after pre-processing with $k_u$, we have $n$ users, $n$ examples each belonging to a different user, and 1 example $e$ shared by every user. We want to know whether including / removing $e$ increases or decreases the expected batch size. In the case we exclude $e$, we just have $n$ copies of the individualized $b$-min-sep setting, and from prior analysis we know that on average, each example participates once in $(b-1) + 1/p$ iterations. If we include $e$, whenever $e$ participates, it means every example can participate in this iteration (since $e$ is user-sharing with every example). If an example $e'$ also happens to participate in this iteration (with probability $p$), the participation of $e$ has no impact on $e'$ since $e'$ will ban itself for $b-1$ iterations anyways. However, if an example $e'$ does not participate (due to not getting sampled, with probability $1-p$), the participation of $e$ bans it for $b-1$ additional iterations, which means we have an average decrease of $\frac{b-1}{(b-1) + 1/p}$ in the batch size. With $n$ other examples, the average decrease in the batch size for every time $e$ is sampled is
\[(1-p) \cdot \frac{b-1}{b-1 + \frac{1}{p}} \cdot n,\]
and the increase in batch size is 1 (from $e$ itself). In a practical training setting where $p \approx 0$ and $1/p \gg b$, we would roughly have that
\[bpk_e < 1\]
means we should include an example. We caution that this is a very rough analysis to determine $k_e$, and we have neglected the graph structure and the fact that getting $e$ sampled (instead of included in the iteration) also bans other examples.

Nevertheless, under our threat model we treat the hypergraph structure as public information, and seek only to protect the contents of the examples. This corresponds to the fixed-graph (multi-attribution user-level) privacy model proposed and studied by \cite{ganesh2025s}, in which preprocessing steps that depend solely on the attribution graph incur no privacy cost. Under this model, it is therefore permissible in practice to experiment with different preprocessing strategies and subsampling parameters—such as varying the edge-level contribution bound $k_e$—prior to any privacy-sensitive training. In particular, one may use our rule-of-thumb analysis to identify a reasonable range of $k_e$ values and then empirically select a setting that yields favorable utility before running any privacy-leaking computation.

At the same time, the interaction between user-level and edge-level contribution constraints suggests that an explicit $k_e$ constraint is often unnecessary once a sufficiently small user-level bound $k_u$ is imposed. Intuitively, bounding the total number of examples attributed to any single user already restricts the feasible overlap patterns among examples, which in turn limits the effective edge-level contributions. This phenomenon is also observed empirically. For example, on the arXiv dataset considered by \cite{ganesh2025s}, the median number of examples per user is $k_u = 2$. When using a subsampling probability chosen to achieve an expected batch size of approximately $128$, we find that imposing additional $k_e$ constraints does not increase the expected batch size beyond what is already achieved by the $k_u$ constraint alone. In such regimes, enforcing an explicit edge-level bound provides little additional benefit and may be safely omitted without sacrificing utility. Nevertheless, in settings with larger user degrees, higher subsampling probabilities, or more skewed attribution graphs, an explicit edge-level constraint can still play an important role, and we include the $k_e$-based preprocessing framework to accommodate such regimes when needed.

\section{Truncated $b$-Min-Sep Sampling}\label{sec:truncated}

Modern ML training pipelines are most efficient when they operate over batches of fixed sizes. For this reason, \cite{chua2024scalable} proposed truncated Poisson sampling, a variant of Poisson sampling that (i) for batches with size $> B_{\max}$, truncates them to a maximum batch size $B_{\max}$ (here, we assume the truncation is done by picking a uniformly random subset of size $B_{\max}$), and (ii) for batches with size $< B_{\max}$, pads them with zero-gradient examples until the size is $B_{\max}$. To extend $b$-min-sep sampling analogously, we simply form a sequence of batches using $b$-min-sep sampling, and then pad or truncate them accordingly afterwards. Note that this means an example which was sampled and then removed from the batch due to truncation is still ineligible to participate in the next $b$ iterations.

Padding does not affect the gradient sum because it is adding zeros, but truncation does affect the gradient sum and can have an adverse impact on privacy. \cite{chua2024scalable} showed how to bound the impact of truncation on DP-SGD, and \cite{ganesh2025tighter} gave a tightening of their result. \cite{ganesh2025tighter} effectively assumes the adversary can receive an indicator for each iteration indicating whether truncation occurred in that round or not. They effectively break the privacy analysis into the usual Poisson sampling analysis when this indicator is false, and a worst-case privacy analysis conditioned on the indicator being true.

To accommodate truncation in our result, we do something similar, but instead assume an even stronger adversary that sees the number of examples (except for the sensitive example) sampled by Poisson sampling in round $i$ prior to padding or truncation, which we denote $r_i$. This is because in the setting of \cite{ganesh2025tighter}, the size of the pool of examples available in each round (except for the sensitive example) is fixed at $|D| - 1$, so the distribution of $r_i \sim \mathrm{Binom}(|D|-1, p)$ is fixed and an upper-bounding privacy analysis can be defined as a function of this distribution. In our setting the size of this pool varies, and in turn the distribution of $r_i$ changes from round-to-round. In particular, the privacy guarantee of the distribution of $r_i \sim \mathrm{Binom}(k, p)$ is not monotonic in $k$, and it is challenging to determine which value of $k$ is the worst-case for the privacy analysis. By releasing $r_i$ to the adversary and calculating the worst-case privacy loss conditioned on this fixed $r_i$, we remove the need to determine the value of $k$ inducing the worst-case distribution over $r_i$. 

\subsection{Privacy Analysis}

We first state a straightforward generalization of the ``vector-to-scalar'' reduction of Lemma 4.5 in ~\cite{choquette2023privacy}. 

\begin{lemma}\label{lem:2mixturereduction}
    Let $\bfa_1, \bfa_2, \ldots \bfa_k, \bfb_1, \bfb_2, \ldots \bfb_\ell \in \mathbb{R}^{n}$, and $\bfa_1', \bfa_2', \ldots \bfa_k', \bfb_1', \bfb_2', \ldots \bfb_\ell' \in \mathbb{R}^{n\times p}$ be such that for any $i, j$ $\ltwo{\bfa'_i[j, :]} \leq \bfa_i[j]$ (resp. $\ltwo{\bfb'_i[j, :]} \leq \bfb_i[j]$). Then for any $p_1, p_2, \ldots p_k, q_1, q_2, \ldots, q_l \in [0, 1]$ such that $\sum_i p_i = \sum_i q_i = 1$ and $\alpha, \sigma \geq 0$ we have:

    \[H_\alpha\left(\sum_{i \in [k]} p_i \calN(\bfa'_i, \sigma^2 \mathbb{I}), \sum_{j \in [\ell]} q_j \calN(\bfb'_j, \sigma^2 \mathbb{I})\right) \leq H_\alpha\left(\sum_{i \in [k]} p_i \calN(\bfa_i, \sigma^2 \mathbb{I}), \sum_{j \in [\ell]} q_j \calN(-\bfb_j, \sigma^2 \mathbb{I})\right)\]

    and 

    \[H_\alpha\left(\sum_{j \in [\ell]} q_j \calN(\bfb'_j, \sigma^2 \mathbb{I}), \sum_{i \in [k]} p_i \calN(\bfa'_i, \sigma^2 \mathbb{I})\right) \leq H_\alpha\left(\sum_{j \in [\ell]} q_j \calN(-\bfb_j, \sigma^2 \mathbb{I}), \sum_{i \in [k]} p_i \calN(\bfa_i, \sigma^2 \mathbb{I})\right)\]

Furthermore, this holds even if the rows of $\bfa'_i, \bfb'_j$ are chosen adaptively while the $\bfa_i, \bfb_j$ are fixed, i.e. an adversary trying to classify a sample from one of $\sum_{i \in [k]} p_i \calN(\bfa'_i, \sigma^2 \mathbb{I}), \sum_{j \in [\ell]} q_j \calN(\bfb'_j, \sigma^2 \mathbb{I})$ gets to decide the $j$th row of all $\bfa_i', \bfb'_j$ after seeing the first $j-1$ rows of the sample, as long as the bounds on the row norms holds of each $\bfa_i', \bfb'_j$. 
\end{lemma}
\begin{proof}
    The proof is effectively the same as Lemma 4.5 of \cite{choquette2023privacy}, except instead of the observation that $\sum_i p_i \exp(\ltwo{\bfc_i}x)$ is continuous, increasing in $x$, and has range $(0, \infty)$ we make the same observation for $\frac{\sum_i p_i \exp(\ltwo{\bfa_i}x}{\sum_j q_j \exp(-\ltwo{\bfb_j}x)}$.
\end{proof}

By combining \cref{lem:2mixturereduction} with the reduction of \cite{ganesh2025tighter} to map to two worst-case distributions, conditioned on a given set of $\{r_i\}$ values, we get the following worst-case pair:

\begin{corollary}
For DP-MF truncated $b$-min-sep sampling, under the zero-out adjacency\footnote{By applying the equivalent reduction from \cite{ganesh2025tighter} we can derive the analogous pair for e.g. the replace-one adjacency.} $P = \calN(\bfC \bfx, \sigma^2\mathbb{I}), Q = \calN(\bfC \bfy, \sigma^2 \mathbb{I})$ is a dominating pair (i.e., it suffices to compute $H_\alpha$ between these two distributions) where $\bfx, \bfy$ are random variables that are functions of the truncated $b$-min-sep sampling process defined as follows:

\[\bfx(i) = \left\{\begin{array}{ll}
     2 & \text{if the sensitive example is sampled in round }i\text{, truncation occurs, and it survives truncation}\\
      1 & \text{if the sensitive example is sampled in round }i\text{, and no truncation occurs}\\
      0 & otherwise
\end{array}\right\} \]
\[\bfy(i) = \left\{\begin{array}{ll}
     -1 & \text{if the sensitive example is sampled in round }i\text{, truncation occurs, and it survives truncation}\\
      0 & otherwise
\end{array}\right\} \]
\end{corollary}

By the post-processing property, assuming an adversary who can additionally see $\{r_i\}$ (which has the same distribution for two adjacent datasets) in our analysis gives an upper bound on the privacy parameters. We now describe how to do Monte Carlo accounting for this $P, Q$ pair (augmented with the $\{r_i\}$ values). Sampling remains straightforward: we first sample $\{r_i\}$ according to the distribution induced by $b$-min-sep sampling (with or without a warm-start). We also sample the rounds the sensitive example participates in under $b$-min-sep sampling. Finally in rounds where $r_i \geq B_{\max}$, the sensitive example survives truncation w.p. $\frac{B}{r_i + 1}$ (conditioned on $r_i$). Hence we can sample the joint random variables $(\bfx, \{r_i\})$ and $(\bfy,\{r_i\})$ efficiently.

Throughout the rest of this argument we condition on $\{r_i\}$ since we assume it is publicly available to the adversary. For evaluating the privacy loss of a given sample $\bfy$, i.e. computing $\frac{P(\bfy)}{Q(\bfy)}$, a recursion in the style of \eqref{eq:mainrecursion} breaks down because the denominator is no longer a single Gaussian. However, there is an easy fix: let $R$ be the mean-zero Gaussian with variance $\sigma^2$. We can compute $\frac{P(\bfy)}{R(\bfy)}$ and $\frac{Q(\bfy)}{R(\bfy)}$ efficiently using a similar recursion to \eqref{eq:mainrecursion}, and then take the ratio to get $\frac{P(\bfy)}{Q(\bfy)}$. Namely, for computing $\frac{P(\bfy)}{R(\bfy)}$ for a given sample, if $r_i < B_{\max}$ we can apply \eqref{eq:mainrecursion} as before, and if $r_i \geq B_{\max}$ we have:

\[f_i(\bfy) = (1 - p) \cdot f_{i+1}(\bfy)
+ p \cdot \frac{\frac{B_{\max}}{r_i + 1}\mathcal{N}(2c_i, \sigma^2 I) + \frac{r_i + 1 - B_{\max}}{r_i + 1}\mathcal{N}(\mathbf{0}, \sigma^2 I)}{\mathcal{N}(\mathbf{0}, \sigma^2 I)}(\bfy_i, \ldots, \bfy_{i+(b-1)}) \cdot f_{i+b}(\bfy),\]

with the same boundary conditions. For $\ln \frac{Q(\bfy)}{R(\bfy)}$ if $r_i < B_{\max}$ we have:

\[f_i(\bfy) = (1 - p) \cdot f_{i+1}(\bfy)
+ p \cdot f_{i+b}(\bfy)\]

and if $r_i \geq B$ we have:

\[f_i(\bfy) = (1 - p) \cdot f_{i+1}(\bfy)
+ p \cdot \frac{\frac{B_{\max}}{r_i + 1}\mathcal{N}(-c_i, \sigma^2 I) + \frac{r_i + 1 - B_{\max}}{r_i + 1}\mathcal{N}(\mathbf{0}, \sigma^2 I)}{\mathcal{N}(\mathbf{0}, \sigma^2 I)}(\bfy_i, \ldots, \bfy_{i+(b-1)}) \cdot f_{i+b}(\bfy).\]

A proof that these recursions are correct for evaluating $\ln P/R$ and $\ln Q/R$ is analogous to the proof of \cref{thm:dp-correctness}. An important subtletly is that in the warm-start case, one must compute $\frac{f_1(\bfy) + p \sum_{i = 2}^{b} f_i(\bfy)}{1 + (b-1)p}$ separately for $P/R, Q/R$ \textit{before} taking the ratio, rather than taking the ratio and then averaging these values.

\end{document}